\pdfoutput=1
\documentclass[USenglish,twocolumn]{article}
\usepackage[ruled,vlined]{algorithm2e}
\usepackage[top=1in, bottom=1in, left=0.9in, right=0.9in]{geometry}
\usepackage{caption}
\usepackage{todonotes}
\usepackage{subcaption}
\usepackage[intlimits]{amsmath}
\usepackage{amsthm}
\usepackage{hyperref}

\SetKwInput{KwInput}{Input}
\SetKwInput{KwOutput}{Output}
\SetKwInput{KwInit}{Initialization}
\newcommand{\NP}{$\mathcal{\mathbb{NP}}$}
\usepackage{amssymb}
\usepackage[utf8]{inputenc}
\usepackage{xcolor}
\usepackage{svg}
\usepackage{authblk}
\newtheorem{theorem}{Theorem}[section]

\newtheorem{proposition}[theorem]{Proposition}

\newtheorem{corollary}[theorem]{Corollary}

\newtheorem{definition}[theorem]{Definition}
 
\begin{document}
  \author[1,*]{Rajagopal Venkatesaramani}
  \author[3]{Zhiyu Wan}
\author[2,3,4]{Bradley A. Malin}
\author[1]{Yevgeniy Vorobeychik}
\affil[1]{Washington University in St. Louis, Computer Science and Engineering, St. Louis, 63130}
\affil[2]{Vanderbilt University, Electrical Engineering and Computer Science, Nashville, 37232}
\affil[3]{Vanderbilt University Medical Center, Biomedical Informatics, Nashville, 37203}
\affil[4]{Vanderbilt University Medical Center, Biostatistics, Nashville, 37203}

\affil[*]{rajagopal@wustl.edu}

  \title{\huge Defending Against Membership Inference Attacks on Beacon Services}
  
\date{}
\maketitle
\begin{abstract}
{Large genomic datasets are now created through numerous activities, including  recreational genealogical investigations, biomedical research, and clinical care. At the same time, genomic data has become valuable for reuse beyond their initial point of collection, but privacy concerns often hinder access.
Over the past several years, Beacon services have emerged to broaden accessibility to such data. These services enable users to query for the presence of a particular minor allele in a private dataset, information that can help care providers determine if genomic variation is spurious or has some known clinical indication. However, various studies have shown that even this limited access model can leak if individuals are members in the underlying dataset. Several approaches for mitigating this vulnerability have been proposed, but they are limited in that they 1) typically rely on heuristics and 2) offer probabilistic privacy guarantees, but neglect utility.
In this paper, we present a novel algorithmic framework to ensure privacy in a Beacon service setting with a minimal number of query response flips (e.g., changing a positive response to a negative). Specifically, we represent this problem as combinatorial optimization in both the batch setting (where queries arrive all at once), as well as the online setting (where queries arrive sequentially). The former setting has been the primary focus in prior literature, whereas real Beacons allow sequential queries, motivating the latter investigation.
We present principled algorithms in this framework with both privacy and, in some cases, worst-case utility guarantees. Moreover, through an extensive experimental evaluation, we show that the proposed approaches significantly outperform the state of the art in terms of privacy and utility.}
\end{abstract}

\section{Introduction}
Genomic sequencing has become sufficiently cheap to support a wide range of services in the clinical and biomedical research domains, as well as for recreational consumers.
As a result, the creation of large databases of genome sequences has become common.  However, not all organizations have access to the same information, such that there is a need to make such information more widely available.
While the sharing of such data has the potential to stimulate further scientific and clinical advances, it also introduces privacy risks. For example, healthcare organizations harboring such data may promise their patients that they will not disclose information that can be tied back to them. This balance between privacy and the value of shared data (commonly referred to as \emph{utility}) has led to the creation of genomic data sharing services to reveal limited amounts of genomic information; for example, by sharing only summary statistics~\cite{rodriguez2013complexities,sankararaman2009genomic,wan2017expanding,weil2013nci}.

Over the past several years, the \emph{Beacon service}, which is promoted by the Global Alliance for Genomics and Health (GA4GH), has been increasingly adopted.  This service enables a user to query for the presence of particular minor alleles in an underlying private genomic dataset.
While exposing such limited information may appear safe, it has been shown to be vulnerable to membership inference attacks because it allows users to issue queries for every region of the genome~\cite{shringarpure2015privacy,wan2017controlling}.
These attacks assume that the attacker  knows the genome of the target, and leverage a statistical test, often in the form of a likelihood ratio,  that couples this information with the Beacon response to a collection of queries, to determine whether the target is a member of the Beacon dataset.
The resulting membership inference can, in turn, reveal sensitive information about the individual, such as their health status, that membership in this dataset entails, or simply be in violation of the privacy promises made to the dataset constituents when the data was collected.

A common approach to mitigate privacy risks in Beacon services is to flip the values in a subset of the query responses~\cite{cho2020privacy,raisaro2017addressing,wan2017controlling}; for example, responding that a particular allele is absent when, in fact, it is present in the dataset. However, not all methods offer privacy guarantees, and when they do, they are often probabilistic, as is the case for those that are based on statistical perturbations, such as differential privacy~\cite{al2017aftermath,cho2020privacy,dwork2006calibrating}.
Similarly, while minimizing the number of flipped queries is a standard measure of utility, no prior approaches offer formal optimality guarantees.

We introduce a novel framework for preserving privacy in the context of membership inference attacks on Beacon services. We consider both a \emph{batch} Beacon setting, where all queries by a given party are specified at once, and an \emph{online} Beacon setting, in which queries arrive sequentially.
The former has been the primary focus of prior Beacon privacy analyses~\cite{cho2020privacy,raisaro2017addressing,shringarpure2015privacy,wan2017controlling}, whereas the latter is more akin to the way that Beacon services are actually run in practice (e.g., \url{https://beacon-network.org/}).
In addition, we consider two threat models: the first, and most common in prior literature, involves an \emph{a priori} fixed threshold used by an attacker in the likelihood ratio test. 
The second is adaptive in the sense that it takes as input Beacon query responses \emph{after the flipping strategy is implemented}, along with a secondary genomic dataset, to adaptively identify a threshold that separates those in the Beacon dataset from those who are not.
The former threat model captures adversaries that are more opportunistic (e.g., a parent trying to determine sensitive information about a child), whereas the latter models stronger, highly informed, adversaries who are trying to systematically harvest data, e.g., to sell to others.

We present a series of algorithmic approaches in each of the aforementioned privacy settings. 
Our strongest results are in the batch setting with fixed-threshold attacks, where we show that in the important special case of very small sequencing rates, we can obtain \emph{both} privacy and a provable 
worst-case approximation of the minimum number of queries to flip by drawing a connection to the \emph{set cover} problem.
We further provide principled algorithmic approaches for the general problems in the batch setting, for both fixed-threshold and adaptive threat models, all of which guarantee privacy under these threat models.
Moreover, we present effective algorithms for preserving privacy while minimizing the number of flipped queries in the online setting.
Finally, through extensive experiments we demonstrate that the proposed approaches significantly outperform the state-of-the-art, including those based on differential privacy, in terms of utility (minimizing the number of flipped queries) when privacy can be guaranteed, and in either privacy, utility, or both when complete privacy cannot be achieved for all members of the Beacon dataset.

\section{Preliminaries}
\subsection{The Beacon Services}
A \emph{Beacon} is a web service that responds to queries about the presence/absence of a specific allele (say, nucleotide A) at a particular position (e.g., position 121,212,028) and on a particular chromosome (say, chromosome 10) for any genomes in the database~\cite{Knoppers2014,Torres2016}.
Such queries are only meaningful when there is variation of alleles in the overall population, and the positions that exhibit such variation are commonly known as single nucleotide variants (SNVs).
Thus, we can equivalently say that the Beacon service responds to queries about the existence of a particular SNV.

Suppose that the Beacon service (or simply the
\emph{Beacon}) responds to queries pertaining to a collection of $m$ SNVs, which we index by an integer $j \in \{1,\ldots,m\}$.
An SNV for each individual $i$ actually contains two alleles, one from each parent, but to simplify the discussion, we encode each SNV $j$ as a binary value $d_{ij}$ that indicates the presence ($d_{ij} = 1$) or the absence  ($d_{ij} = 0)$ of the minor allele.
Thus, we can represent an individual $i$ as a binary vector $d_i = \{d_{i1},\ldots,d_{ij},\ldots,d_{im}\}$.

We say that the individuals who are a part of this dataset are \emph{in the Beacon}, contrasting with those not in this dataset, who are referred to as \emph{not in the Beacon}.
If an SNV at position $j$ is queried, the Beacon returns a response $x_j = 1$ if at least one individual $i$ in the Beacon has the minor allele, and 0 otherwise.

\subsection{The Likelihood Ratio Test (LRT) Attack}
Let $B$ be the set of $n$ individuals in the Beacon, and $S$ be the set of $m$ SNV positions that can be queried in the Beacon (in other words, we can view Beacon queries simply as integer indices corresponding to SNV positions).
Define $\delta$ as the genomic sequencing error, and let $f^j$ be the alternative (minor) allele frequency (AAF) in the population for each position $j$. 
Let $D^j_n$ denote the probability that no individual in the beacon has a minor allele at position $j$, which is calculated as \[D^j_n = (1-f^j)^{2n}\]
(recall that each SNV has 2 alleles for each individual, hence $2n$).

We begin with the well-known membership inference attack on the Beacon service by Shringarpure and Bustamante~\cite{shringarpure2015privacy} (henceforth, SB), with the additional assumption that the attacker knows AAF for each position $j$~\cite{wan2017controlling}.
In this attack, the attacker first submits a collection of queries $Q$ to the Beacon service, and then uses these to calculate the likelihood ratio test (LRT) statistic for each individual in the attacker's target set $T$ (i.e., the set of individuals whose membership in the Beacon the attacker wishes to infer).
Specifically, given an individual $i \in T$, a set of queries $Q$, and the vector of query responses $x$, the LRT statistic is 
\begin{equation}
\begin{split}
  L_i(Q,x) = \sum_{j \in Q} d_{ij}\left(x_j\log\frac{1-D^j_n}{1-\delta D^j_{n-1}} + \right. \\ \left. (1-x_j)\log\frac{D^j_n}{\delta D^j_{n-1}}\right).
 \end{split}
 \label{eq:LRT}
\end{equation}
Finally, the attacker claims that an individual $i \in T$ is in the Beacon $B$ when $L_i(Q,x) < \theta$.
The choice of $\theta$ reflects the adversary's preferred balance between precision and recall of the membership inference attack.

\section{Threat Models}
\label{S:threat}

Our threat models for membership inference attacks on the Beacon service are anchored in the SB likelihood ratio test (LRT) attack described above.
However, the LRT attack leaves open three questions. First, what is the attacker's target set $T$? Second, how does the attacker arrive at the choice of a threshold $\theta$ to determine which Beacon membership claims are made? Third, what is the set of queries used in the attack?
Since we do not know \emph{a priori} which individuals will be targeted, we make the worst-case assumption that $B \subseteq T$. In other words, the attacker targets everyone in the Beacon, along with possibly others.
Our threat modeling leads to several variants of the LRT attack along the remaining two dimensions.

\subsection{Choosing the Inference Threshold}

We consider two approaches that an adversary may use to determine when to make a membership inference claim about an individual: \emph{fixed-threshold} and \emph{adaptive} attacks.
In \emph{fixed-threshold} attacks, an adversary uses a predefined threshold $\theta$, which is fixed for the inference attack.
This is a common threat model in the literature~\cite{Hagestedt19,wan2017controlling} and reflects an opportunistic attacker who initially uses a private dataset to simulate LRT attacks by splitting individuals into those in a simulated Beacon and those who are not.
These offline simulations are then used to identify the $\theta$ that best balances precision and recall with respect to attacker's preferences about these.
The practical consequence of assuming a fixed $\theta$ at the time of attack is that $\theta$ is not adjusted based on query responses.

Not considering queries in determining the threshold $\theta$ is consequential once we consider defensive measures that modify query responses.
For example, if modified query responses preserve a clear separation in LRT statistics between individuals who are in and not in the Beacon, a simple clustering of the statistics would enable the attacker to effectively identify those in the Beacon.
Consequently, we additionally consider a stronger  \emph{adaptive} threat model that sets $\theta(x_Q)$ as a function of the responses $x_Q$ to queries $Q$, but with the aim of limiting the false positive rate for any membership claims to be at most $\alpha$.
This \emph{adaptive} threat model requires that  attacker can set the threshold in precisely the right place based on actual queries $Q$. However, this can effectively be accomplished with the aid of simulation experiments using a private dataset $D$ (now, simulating Beacon queries that implement the defensive measures).
This threat model captures highly informed attackers, for example, those who do systematic harvesting of sensitive data for profit.

\subsection{Query Process}
In this work, we distinguish between three mechanisms of query access that can be provisioned to a Beacon: 1) batch query, 2) unauthenticated online query access, and 3) authenticated online query access.
In the \emph{batch} setting, the attacker is assumed to query \emph{all} $m$ SNP positions effectively simultaneously.
In a sense, this is the most favorable setting for the adversary, as it provides maximum information for making membership claim decisions.
It is also the setting that has received much of the attention in prior literature~\cite{Hagestedt19,wan2017controlling}.
However, typical Beacons in practice (such as the service provided by the GA4GH Beacon Network) can be queried \emph{sequentially}, and we thus need to ensure that privacy is guaranteed even in such settings.
Consequently, we also consider two online settings.
The \emph{unauthenticated online} setting assumes that the attacker can submit an arbitrary subset of queries.
This is because if queries are not connected to a particular identity, there is no way for the Beacon to know which queries have been made by the same individual in the past.
The public Beacon Network is an example of this situation.
The \emph{authenticated online} setting, by contrast, assumes that we can keep track of all the past queries by each individual, including the potential adversary.
This entails only allowing registered and verified users access (we assume no collusion), and allows (as we show in the experiments) privacy guarantees with significantly higher utility for such users.

We now make an observation that enables us to uniformly talk about the three variants of the query process above: note that the sole mathematical distinction between them is the set of queries $Q$ that we are concerned about.
In the \emph{batch} setting, $Q = S$, the set of all queries: that is, the LRT statistics of relevance for the purposes of membership inference attack is computed with respect to the set of all possible queries.
In the online authenticated online setting, $Q$ is the set of all past queries, together with the current query $j$, since we are concerned that an individual $i$ may be identifiable as soon as $L_i(Q)$ drops below $\theta$ (or $\theta(Q)$, in the adaptive model).
Finally, in the unauthenticated online setting, since we do not know which set of queries have been asked by the adversary, or what the target set is, we make the worst-case assumption that the adversary makes \emph{most identifying queries} for each individual $i$, that is, the set of queries $Q_i$ is specific to each individual in the Beacon and minimizes $L_i(Q_i) - \theta(Q_i)$ for each individual.

Since the choices for adversarial queries $Q$ (or $Q_i$ in the unauthenticated online case) are thus isomorphic with the particular query process in the threat model, we henceforth simply focus on two aspects of threat models: 1) the choice of query set $Q$ and 2) whether $\theta(Q)$ depends on the responses to $Q$.

\section{Privacy and Utility Goals}
We  now formalize the goals for protecting Beacon service privacy.
Following the framework of the 2016 iDash \emph{Practical Protection of Genomic Data Sharing through Beacon Services} challenge~\cite{wan2017controlling}, the primary means to protect the data allow the Beacon service to flip the answer to a subset of possible SNV queries $F \subseteq S$.
We encode the choice of which responses to flip as a binary vector $y = \{y_1,\ldots,y_m\}$, where $y_j = 1$ implies that the response to query $j \in S$ is flipped, and $y_j = 0$ means that the query answer is unchanged.
We define $x_Q(F)$ as the vector of Beacon  responses to queries $Q$ when the set $F$ of responses is flipped.

Our privacy goal is to ensure that the privacy is preserved \emph{for all individuals in the Beacon}.
Let $L_i(Q_i,x_{Q_i}(F))$ be the LRT statistic  and $\theta(x_{Q_i}(F))$ the threshold after we flip the set $F$ of query responses.
Formally, we wish to guarantee that
\begin{equation}
\label{E:privacy}
\forall i \in B, \quad L_i(Q_i,x_{Q_i}(F)) - \theta(x_{Q_i}(F)) \ge 0.
\end{equation}
In the case of a fixed-threshold attack, $\theta(x_{Q_i}(F))$ is a constant independent of $x$; in the batch setting, $Q_i = S$ for all $i$; and in the authenticated online setting, $Q_i = Q$ for all $i$, where $Q$ is the set of queries made thus far by the authenticated user.

Clearly, we can preserve privacy by simply shutting down the Beacon service.
However, there is value in genomic data sharing, and it is this value that has motivated creative ideas for sharing it in a privacy respectful manner.
Our broader goal, therefore, is thus to achieve privacy, as defined by Equation~\eqref{E:privacy}, with a minimal impact on utility, which in this context means minimizing the number of query responses that are flipped.
The resulting optimization problem, which we refer to as the \textsc{Beacon-Privacy-Problem} can be formalized as follows:
\begin{equation}
\begin{aligned}
\label{E:prob}
\min_{F \subseteq S} |F| \ \mathrm{subject\ to:}\quad\quad\quad\\ L_i(Q_i,x_{Q_i}(F)) - \theta(x_{Q_i}(F))\ge 0 \ \forall i \in B.
\end{aligned}
\end{equation}
In addition, we aim to solve Problem~\ref{E:privacy} effectively and efficiently for each of the threat model settings described in Section~\ref{S:threat}.

To begin, we make a useful structural observation that significantly limits the set of query responses to be considered for flips: we would never want to flip responses from 0 to 1.
We formalize this in the following proposition.
\begin{proposition}
\label{lemma_flip}
Suppose that $x_j = 0$ given the Beacon dataset.  Then $y_j = 1$ can never increase the LRT statistic for any individual $i \in B$, provided sampling error $\delta < \frac{D^j_n}{D^j_{n-1}}$ for all $j$.
\end{proposition}
Due to space constraints, we defer the proof to Appendix~\ref{A:a1}.
Since a privacy violation corresponds to LRT statistics for at least one individual in the Beacon being too small, our goal is necessarily to increase these scores until privacy is guaranteed for all individuals in the Beacon.
Henceforth, we assume that $\delta < 0.25 \le  \frac{D^j_n}{D^j_{n-1}} = (1-f^j)^2$ (since $f^j < 0.5$ for all $j$).
Consequently, Proposition~\ref{lemma_flip} implies that flipping a 0 response to a 1 is counterproductive, and we need not consider it as a possible solution to Problem~\ref{E:prob}.
Then, without loss of generality, we can assume that our consideration set $S$ includes only the query responses which are initially 1.

Next, we show that even in a very restricted special case, the problem of minimizing the number of flips in order to guarantee the privacy of all individuals in the Beacon is \NP-Hard.
First, we define a decision version of Problem~\ref{E:prob}, which we refer to as \textsc{Beacon-Privacy-D}.
\begin{definition}[\textsc{Beacon-Privacy-D}]
\textbf{Input:} A collection of individuals $i \in B$ with genomic information $D$, and induced Beacon query responses $x$; a constant $k$.
\textbf{Question:} Can we flip a subset $F$ of query responses where~$|F| \le k$ and $L_i(Q_i,x_{Q_i}(F)) - \theta(x_{Q_i}(F))\ge 0$?
\end{definition}
We reduce from the \textsc{Set Cover} problem.
\begin{definition}[\textsc{Set Cover}]
\textbf{Input:} A universe $U$ of elements, and a collection of sets $R = \{R_1,\ldots,R_n\}$ with $R_j \subseteq U$ and $\cup_j R_j = U$; a constant $k$.
\textbf{Question:} Is there a subset $T \subseteq R$ of sets such that $U = \cup_{t \in T} t$ and $|T| \le k$.
\end{definition}

\begin{theorem}
\label{T:hard}
\textsc{Beacon-Privacy-D} is \NP-Complete even if $\delta = 0$ and $\theta(x_{Q_i}(F))$ is a constant.
\end{theorem}
The proof is deferred to Appendix \ref{A:a2} due to space constraints.

\section{The Batch Setting}

We begin by considering the \emph{batch} setting in which the adversary submits a set of queries $Q$ all at once, where $Q \ne \emptyset$.
This provides the building blocks for all the query settings in our threat model, both what we called \emph{batch query} setting with $Q = S$ above, and for the online settings.

Recall that a binary vector $x$ corresponds to the \emph{true} query responses, while $y$ represents whether responses have been flipped.
Let $Q_1$ be the subset of queries $Q$ with $x_j = 1$ and $Q_0$ be the subset with $x_j = 0$.
Let $A_j = \log \frac{1-D^j_n}{1-\delta D^j_{n-1}}$ and $B_j = \log \frac{D^j_n}{\delta D^j_{n-1}}$. 
We can then rewrite $L_i(Q,x)$ as follows:
\[
  L_i(Q,x) = \sum_{j \in Q_1} d_{ij}A_j + \sum_{j \in Q_0}d_{ij}B_j.
\]
Moreover, by Proposition~\ref{lemma_flip}, we never flip queries in $Q_0$, which means that for our purposes the second term above is a constant.
Now, if we apply the query flip strategy $y$, the resulting LRT statistic, which we denote by $L_i(Q,x,y)$, becomes
\begin{align*}
  L_i(Q,x,y) = \sum_{j \in Q_1}  d_{ij}((1-y_j)A_j + y_j B_j) + \sum_{j \in Q_0}d_{ij}B_j\\
=\sum_{j \in Q_1} y_j d_{ij}(B_j-A_j) + \sum_{j \in Q_1} d_{ij} A_j + \sum_{j \in Q_0}d_{ij}B_j.
\end{align*}
Define $\Delta_{ij} = d_{ij}(B_j-A_j)$ and $\eta_i = \sum_{j \in Q_1}  d_{ij} A_j + \sum_{j \in Q_0}d_{ij}B_j$.
Then
\begin{align}
\label{E:lrt_y}
    L_i(Q,x,y) = \sum_{j \in Q_1} \Delta_{ij} y_j + \eta_i.
\end{align}
Note that $\eta_i$ is actually also a function of the set of queries $Q$.
For the remainder of this section this will not be important and so we omit this dependence.
However, this becomes important in the online setting below.

Next, we consider approaches to solve \textsc{Beacon-Privacy-Problem} first in the fixed-threshold and subsequently in the adaptive attacks.

\subsection{Fixed-Threshold Attacks}
\label{S:ft}

We begin by presenting an integer linear programming (ILP) approach for optimally solving the general \textsc{Beacon-Privacy-Problem} in the batch setting with fixed-threshold attacks.
This is a straightforward consequence of the problem structure that we had already derived above.
First, note that we wish to minimize the number of flips, which is equivalent to minimizing the number of $1$s in $y$.
Second, note that the privacy constraint is $L_i(Q,x,y) \ge \theta$ which is linear in $y$.
Consequently, the following ILP solves the \textsc{Beacon-Privacy-Problem}:
\begin{align}
\label{E:ilp}    
\min_{y \in \{0,1\}^m} \sum_j y_j~\mathrm{subject\ to:} \nonumber
\\\sum_{j \in Q_1} \Delta_{ij} y_j + \eta_i \ge \theta \ \forall \ i \in B. 
\end{align}

This ILP has worst-case exponential running time, and consequently will have trouble scaling to large problems that include thousands of individuals and millions of SNPs.
We proceed to address the scalability in two ways: first, we identify important special cases which either enable the ILP solvers to leverage problem structure, or admit an approximation algorithm with worst-case guarantees; second, we present two greedy algorithms for solving the general variant of the \textsc{Beacon-Privacy-Problem}.
The key property of \emph{all} the solutions we propose is that they \emph{satisfy the privacy constraints by construction}.

\subsubsection{Small Sequencing Error Rates}

Genomic sequencing error rates $\delta$ are often quite small, on the order of $10^{-6}$.
We now show that for sufficiently small sequencing error rates (with $\delta = 0$ a special case), we can represent the \textsc{Beacon-Privacy-Problem} as a \textsc{Minimum Set Cover} instance.
This, in turn, implies that we can solve our problem using a greedy algorithm with a logarithmic worst-case approximation guarantee.

Generally speaking, for a sufficiently small $\delta$, the terms $B_j$ will be extremely large for any query $j$ that we may choose to flip from a 1 to a 0 and, in particular, will be much larger than $A_j$.
Thus, for every individual $i$ in the Beacon, flipping any $j \in Q_1$ will result in a very large increase in $\Delta_{ij} = d_{ij} (B_j - A_j)$.
This increase will, indeed, be so large as to guarantee that $L_i(Q,x,y) \ge \theta$.
As a consequence, it will suffice to flip \emph{any} query with $d_{ij} = 1$ for $i$ to no longer be categorized as in the Beacon by the attack.
Of course, just how small $\delta$ needs to be for this to work will be a function of both the problem parameters $D_n^j$ and $D_{n-1}^j$, as well as $\theta$.
We emphasize that this line of reasoning is \emph{specific to the fixed-threshold attack} model; the issue is far more subtle in adaptive attacks (Section~\ref{S:adaptive}).
Next, we make this premise precise.

For each individual $i \in B$, define $P_i = \{j \in Q_1 | d_{ij} = 1\}$; in other words, $P_i$ is the set of all queries $j$ for which a) $x_j = 1$ and b) the individual $i$ actually has the associated alternate allele for query $j$, i.e., $d_{ij} = 1$.
We now provide a sufficient condition on $\delta$ such that we can flip any query in $P_i$ for each $i \in B$ to guarantee privacy under fixed-threshold attacks.

\begin{definition}
A set of queries $F$ that have been chosen to flip is a \emph{Beacon-Cover} if $\forall i \in B, F \cap P_i \ne \varnothing$.
\end{definition}
In words, $F$ is a \emph{Beacon-Cover} if each individual in the Beacon is covered by some query flipped in $F$ which is also in $P_i$.
We now define two additional components to the notation that will be useful throughout our analysis.
First, define
\[
D_n = \min_{j \in Q_1} \log(D_n^j/(1-D_n^j)).
\]
Second, define
\begin{align*}
\eta=\min_i\left(\sum_{j \in Q_1} d_{ij} \log(1-D_n^j) + \right.
\\ \left. \sum_{j \in Q_0} d_{ij} \log\frac{D_n^j}{0.25 D_{n-1}^j} \right)
\end{align*}
The following proposition presents a bound on $\delta$ that ensures that a \textsc{Beacon-Cover} guarantees privacy against fixed-threshold attacks.

\begin{theorem}
\label{T:bc}
Suppose that $\delta \le \frac{1}{1+e^{\theta - \eta - D_n}}$.
Then, if $F$ is a \emph{Beacon-Cover}, it guarantees privacy of all $i \in B$ against fixed-threshold attacks with threshold $\theta$.
\end{theorem}
The full proof of this result is provided in Appendix~\ref{A:a3}.
The benefit of Theorem~\ref{T:bc} is that it suffices for $F$ to ``cover'' each individual in the sense that for every individual $i$ in the Beacon, there is at least one flipped query in $F$ that suffices to ensure that the score $L_i(Q,x,y) \ge \theta$, that is, to ensure that $i$'s privacy is preserved under the fixed-threshold threat model in the batch setting.
This, in turn, allows us to represent the \textsc{Beacon-Privacy-Problem} as a \textsc{Min Set Cover} instance.
The \textsc{Min Set Cover} problem is the optimization variant of \textsc{Set Cover}, which we now define formally.
\begin{definition}[\textsc{Min Set Cover}]
\textbf{Input:} A universe $U$ of elements, and a collection of sets $R = \{R_1,\ldots,R_n\}$ with $R_j \subseteq U$ and $\cup_j R_j = U$; a constant $k$.
\textbf{Goal:} Minimize $|T|$ over $T \subseteq R$ such that $U = \cup_{t \in T} t$.
\end{definition}
We now show how to represent our problem as an instance of the \textsc{Min Set Cover} problem.
The key advantage of this representation will be a greedy algorithm for solving the \textsc{Beacon-Privacy-Problem} in this setting that yields a logarithmic worst-case approximation guarantee~\cite{Slavik96}.
The representation is similar to the one used in the proof of Theorem~\ref{T:hard} but, of course, is in the opposite direction.
Specifically, we are given a \textsc{Beacon-Privacy-Problem} instance, which we now use to construct a \textsc{Min Set Cover} instance.
Let $U = B$, the set of the individuals in the Beacon, while each $R_j$ corresponds to query $j$, and is comprised of the individuals $i \in B$ whose privacy will be protected if we flip $j$.
Formally, $R_j = \{i \in B | j \in P_i\}$.

Now, we can leverage the greedy algorithm for \textsc{Min Set Cover} to solve our problem.
The greedy algorithm works as follows.
The collection of subsets $T$ is initialized to be empty.
Then, in each iteration we add a subset $S_j$ to $T$ which maximizes the number of elements in $U$ it covers that are not already covered by $T$.
We stop when the entire universe $U$ is covered.
Algorithm~\ref{algo:cover} in Appendix~\ref{A:alg1}, which we refer to as \textsc{Greedy Min Beacon Cover} presents a direct adaptation of this to our problem.
The following is, thus, a direct corollary of Theorem~\ref{T:bc}.
\begin{corollary}
Suppose that $\delta \le \frac{1}{1+e^{\theta - \eta - D_n}}$.
Then Algorithm~\ref{algo:cover} gives an $O(\log(n))$-approximation to the \textsc{Beacon-Privacy-Problem}.
\end{corollary}

\subsubsection{Alternate Allele Frequencies Drawn from the Beta Distribution}
\label{AAF_Beta}
A common assumption in prior literature is that the alternate allele frequencies (AAFs) are drawn from a Beta distribution~\cite{wan2017controlling}, replacing the $D^j_n$ and $D^j_{n-1}$ terms in the LRT score calculation with the expectation over the distribution, which we denote by $\bar{D}_n$ and $\bar{D}_{n-1}$, respectively.
This, in turn, means that $A_j$ and $B_j$ are independent of $j$, and we now denote them by constants $A$ and $B$, respectively.
As a result, $\Delta_{ij} = d_{ij} (B-A)$ and $\eta_i = A\sum_{j \in Q_1} d_{ij} + B\sum_{j \in Q_0} d_{ij}$, and we obtain a simpler expression for the LRT statistic induced by query flips $y$:
\[
L_i(Q,x,y) = (B-A) \sum_{j \in Q_1}d_{ij}y_j + \eta_i.
\]
Consequently, the privacy condition for each $i \in B$ is equivalent to
\[
\sum_{j \in P_i} y_j \ge k_i, \quad \mathrm{where} \quad k_i = \frac{\theta - \eta_i}{B-A}.
\]
Note that, under our assumption,  $\delta < 1/4$, $B-A > 0$.
This has two algorithmic implications.
First, it yields a significantly simpler set of privacy constraints in the integer linear program~\eqref{E:ilp} to obtain the optimal solution to the \textsc{Beacon Privacy Problem}.
Second, we can derive a natural greedy heuristic for this case which generalizes the Greedy \textsc{Min Beacon Cover} algorithm above.

The high-level idea of the greedy heuristic is to iteratively choose a query result to flip that affects the largest number of individuals.
This idea is formalized in Algorithm~\ref{algo:kcover}, which we term \textsc{Greedy $k$-Cover}.

\begin{algorithm}[h]
\SetAlgoLined
\KwInput{A set $B$ of individuals in the Beacon, a subset $P_i$ and a constant $k_i$ for each individual, and a collection of queries $Q$.
}
\KwOutput{Subsets of queries $F \subseteq S$ to flip.}
\KwInit{$F = \emptyset$, $C = \emptyset$.}
  \While{$(B \setminus C) \ne \emptyset$} {
    Set $l = 1, N = -1$.\\
    \For{$j \in (Q \setminus F)$}{
      Set $T_j = \{i \in (B \setminus C) | j \in P_i\}$.\\
      \If{$|T_j| > N$}{
      Set $N = |T_j|$.\\
      Set $l = j$.
      }
      Set $F = F \cup l$.\\
      \For{$i \in (B\setminus C)$}{
         \If{$F \cap P_i \ge k_i$}{
           Set $C = C \cup i$.
         }
      }
    }
 }
 \caption{Greedy $k$-Cover}
 \label{algo:kcover}
\end{algorithm}

\subsubsection{Heuristic Approaches for the General Case}
\label{S:ftgen}
While the two special cases considered above are common, the assumptions in these do not always hold.
On the other hand, the integer programming approach~\eqref{E:ilp} is unlikely to scale to large problems, especially when we have millions of queries to consider.
We now present a general-purpose greedy heuristic that builds on the \textsc{Greedy $k$-Cover} algorithm.
First, observe that in the general setting there is no longer a meaningful notion of "cover", since each query and individual have an associated specific contribution $\Delta_{ij}$.
On the other hand, recall that $\Delta_{ij} = d_{ij} (B_j - A_j)$ and, consequently, if $j \in P_i$, then the marginal impact of flipping query $j$ on the LRT statistic of $i$ only depends on query $j$.
Define $\Delta_j = (B_j - A_j)$, so that $\Delta_{ij} = d_{ij} \Delta_j$.
For any subset of individuals $P \subseteq B$ and query $j$, let $T_j = \{i \in P | j \in P_i\}$ be the set of individuals for whom $j \in P_i$. We can then define the \emph{average marginal contribution} of each query $j$ and population $P$ as \[\bar{\Delta}_j(P) = \frac{|T_j|\Delta_j}{|P|}.\]
\begin{algorithm}[h]
\SetAlgoLined
\KwInput{A set $B$ of individuals in the Beacon, a subset $P_i$ and $\eta_i$ for each individual, marginal contributions $\Delta_j$ for each query, a collection of queries $Q$, and a threshold $\theta$.
}
\KwOutput{Subsets of queries $F \subseteq S$ to flip.}
\KwInit{$F = \emptyset$, $C = \emptyset$.}
  \While{$(B \setminus C) \ne \emptyset$} {
    Set $l = 1, N = -1$.\\
    \For{$j \in (Q \setminus F)$}{
      Set $T_j = \{i \in (B \setminus C) | j \in P_i\}$.\\
      Set $\bar{\Delta}_j = \Delta_j\frac{|T_j|}{|B \setminus C|}$.\\
      \If{$\bar{\Delta}_j > N$}{
      Set $N = \bar{\Delta}_j$.\\
      Set $l = j$.
      }
      Set $F = F \cup l$.\\
       \For{$i \in (B\setminus C)$}{
         \If{$\sum_{j \in F} \Delta_{ij} + \eta_i \ge \theta$}{
           Set $C = C \cup i$.
         }
      }
    }
 }
 \caption{MI Greedy}
 \label{algo:sortave}
\end{algorithm}

The greedy heuristic we propose iteratively chooses a query $j$ to flip with the largest marginal contribution $\bar{\Delta}_j(P)$, where the population $P$ consists of the individuals whose privacy has yet to be guaranteed.
Note that for this heuristic to work reasonably well, it is crucial that $\Delta_j > 0$. This is indeed the case as shown in the proof of Proposition~\ref{lemma_flip} (which implies that $B_j - A_j > 0$) if $\delta < 1/4$.
This means that as we flip query responses, we cannot decrease the LRT score for any individual, and, consequently, any individual $i$ whose privacy is already protected remains protected.
This heuristic is formalized as Algorithm~\ref{algo:sortave}, which we call \textsc{MI-Greedy} (where MI stands for \emph{Marginal Impact}).

\subsection{Adaptive Attacks}
\label{S:adaptive}
Thus far, our attention has centered on the threat model in which the attacker fixes the decision threshold $\theta$ prior to executing any queries, and $\theta$ is independent of queries.
However, since our defense involves the alteration of query responses, an adaptive attacker should make use of query responses in identifying an appropriate threshold.
This means that from the perspective of privacy protection, it is not sufficient to ensure that LRT statistics for all individuals exceed some predefined threshold, but we must actually ensure that the Beacon and non-Beacon populations are \emph{well-mixed} in terms of the respective LRT statistics as calculated \emph{based on the modified query responses}.
We now formalize this idea.

Consider our encoding $y$ of which query responses to flip, and let $\bar{B}$ denote a  set of individual genomes not in the Beacon (e.g., a data sample of these from the general population).
The LRT statistic for each individual $i \in \bar{B}$ can be computed just as for any $i \in B$.
Let us take $K$ individuals from $\bar{B}$ with the lowest LRT statistics, denoting the set of these individuals by $\bar{B}^{(K)}$.
The concrete instantiation of the adaptive threat model then uses the following threshold:
\[
\theta(Q) = \frac{1}{K}\sum_{k \in \bar{B}^{(K)}} \left(\sum_{j \in Q_1} \Delta_{kj} y_j + \eta_k\right).
\]
We can interpret this as representing an attacker's tolerance for false positives. For example, if the distribution of LRT scores is approximately symmetric around the mean, then $K/(2|\bar{B}|)$ is approximately the false positive rate.
As a result, the privacy constraint for each $i \in B$ in the adaptive attack setting becomes
\begin{align*}
\sum_{j \in Q_1} \Delta_{ij} y_j + \eta_i &\ge \frac{1}{K}\sum_{k \in \bar{B}^{(K)}} \left(\sum_{j \in Q_1} \Delta_{kj} y_j + \eta_k\right)\\
&=\sum_{j \in Q_1} \sum_{k \in \bar{B}^{(K)}}\left(\frac{\Delta_{kj}}{K}\right) y_j + \sum_{k \in \bar{B}^{(K)}}\left(\frac{\eta_k}{K}\right).
\end{align*}
Define
\[
\Delta^{(K)}_j = \sum_{k \in \bar{B}^{(K)}} \frac{\Delta_{kj}}{K} \quad \mathrm{and} \quad \eta^{(K)} = \sum_{k \in \bar{B}^{(K)}} \frac{\eta_{k}}{K}.
\]
Rewriting the expression above, we then obtain the following privacy condition for $i \in B$:
\begin{align}
\label{E:adaptiveprivacy}
\sum_{j \in Q_1} (\Delta_{ij} - \Delta_j^{(K)}) y_j + \eta_i \ge \eta^{(K)}.
\end{align}
Finally, by defining $\Delta_{ij}^{(K)} = \Delta_{ij} - \Delta_j^{(K)}$, we can rewrite this in the form identical to Equation~\eqref{E:lrt_y} for the fixed threshold attacks:
\begin{align}
\label{E:lrt_ad_y}
\sum_{j \in Q_1}\Delta_{ij}^{(K)}y_j + \eta_i \ge \eta^{(K)}.
\end{align}

Superficially, this suggests that we can directly apply the methods developed in Section~\ref{S:ft} for privacy-protection against fixed-threshold attacks.
And, indeed, we can directly incorporate the linear privacy constraint~\eqref{E:lrt_ad_y} into the linear integer program~\eqref{E:ilp}.
However, this threat model now breaks the greedy algorithms we previously proposed.
The first reason is that $\delta$ now figures as a part of the threshold $\theta(Q)$ and, consequently, is embedded in $\Delta_{ij}^{(K)}$ in two potentially conflicting ways.
The second (and related) issue is that while a fixed-threshold threat model implied, for $\delta < 0.25$, that $\Delta_{ij} > 0$, this is clearly no longer necessarily the case for $\Delta_{ij} - \Delta_j^{(K)}$.
This has two consequences: 1) greedily adding one query $j$ to the flip set $F$ may actually cause privacy violation of an individual whose privacy constraint was previously satisfied, and 2) the integer linear program~\eqref{E:ilp} may no longer have a feasible solution even though it is feasible for a fixed $\theta$.
In practice, this means that the choice of $K$ cannot be overly conservative.
Moreover, to enable us to directly reuse the general-purpose greedy algorithm from Section~\ref{S:ft} for privacy protection against the adaptive threat model, we only consider flipping queries $j$ for which $\Delta_{ij}^{(K)} \ge 0$ for all $i \in B$ (in our experiments below, that corresponds to 631047 out of 1338843 total SNVs on chromosome 10).

\section{The Online Setting}

Thus far, we assumed that the attacker submits all queries $Q$ all at once, computes LRT statistics, and decides which individuals to make membership claims about.
In practice, however, queries to the Beacon arrive over time, and privacy violations may arise even inadvertently if the attacker is, say, a relative of an individual in the Beacon who happens to observe that a rare collection of minor alleles that their family member possesses is also in the Beacon.
Since individual queries may increase as well as decrease LRT statistics, it may well be the case that queries flipped in anticipation of a batch attack---even with $Q=S$---nevertheless violate privacy for \emph{some} query sequences.
Consequently, in the online setting we need to assure Beacon service privacy \emph{for subsets of queries}.

However, note that, in practice, we may not need to be concerned about \emph{arbitrary} subsets of queries: since it is optimal for an attacker's perspective to make use of all query responses they have observed, we need only guarantee privacy for the subset of queries \emph{submitted by any user thus far}.
That is, of course, if we \emph{know} which queries the user submitted.
This issue of whether or not the Beacon service knows which queries have previously been submitted by a user motivates a natural distinction between two classes of online use settings that we discussed in Section~\ref{S:threat}: \emph{authenticated access}, where the Beacon knows all prior queries (i.e., access \emph{requires} authentication and identity is carefully verified), and \emph{unauthenticated access}, where the Beacon does not have this information.
Next, we formalize the online query setting, and subsequently consider in turn authenticated and unauthenticated access.

\subsection{A Model of the Online Beacon}

The online query setting is characterized by a sequence of $T$ queries $\{q_1,\ldots,q_T\}$, with $q_t \in S$ denoting a $t$th query about a particular SNV (we alternatively refer to this as a query at time $t$, with time here being equivalent to the order in the query sequence).
Similarly, at each point in time, including $t=0$ (i.e., before any queries), the Beacon can decide to flip a subset of query responses $F_t$.
In this setting, the set of all queries flipped is $F = \cup_t F_t$; however, in the online setting we need not flip them all at once.
The reason we may choose to defer flipping a particular query response is that observed queries are informative, and a particular observed query sequence may warrant flipping many fewer responses than, say, a worst-case sequence or the batch of all queries $S$.
There is an additional constraint that we must impose on $F_t$: query responses are \emph{commitments}, in the sense that if at any point $t$ we choose to honestly respond to a query $j$, we must do so in the future; similarly, if we chose to flip the query response, we must do so in the future as well.
Given this constraint, we assume without loss of generality that the query sequence is \emph{non-repeating}, that is, for all $1\le t,t' \le T$, $q_t \ne q_{t'}$ (since future identical queries are responded to exactly as the first time they are encountered).

At time $1\le t\le T$, we have a collection $Q_{t-1}$ of past queries, along with the query $q_t$ that just arrived, resulting in the query set $Q_t$ observed thus far.
A privacy guarantee now entails that privacy of no individual $i 
\in B$ is violated \emph{at any time} $t$.
For a fixed-threshold threat model, this translates into the following privacy condition;
\[
\forall i,t, \quad L_i(Q_t,x) \ge \theta(Q_t).
\]
As we observed in Section~\ref{S:adaptive}, the condition has an analogous form for adaptive attacks.
Since we are in the online setting, we can now choose subsets of queries to flip over time, rather than all at once, we can encode the associated decisions $F_t$ as binary vectors $y_t$, resulting in the following privacy condition:
\[
\forall i,t, \quad L_i(Q_t,x,y_t) \equiv \sum_{j \in Q_{1,t}} \Delta_{ij}y_{j,t} + \eta_i(Q_t) \ge \theta(Q_t),
\]
where we now make it explicit that $\eta_i$ in the modified LRT statistics depends on the query set $Q_t$.

\subsection{Authenticated Access}

The key feature of authenticated access settings that we can leverage is the knowledge at any time $t$ of the prior queries $Q_{t-1}$ as well as the current query $q_t$ to which the Beacon is about to respond (effectively, assuming that there is no collusion among Beacon clients, unlike in the unauthenticated setting below). 
The following proposition makes the intuitive observation that in the authenticated access setting, one never needs to make a decision whether to flip a query or not at time $t$ for any $j \ne q_t$.
\begin{proposition}
\label{P:oneflip}
For any $1\le t \le T$, there is an optimal query flip policy with $F_t \subseteq \{q_t\}$.
\end{proposition}
This follows because if you wish to flip a particular query $j$, you need not \emph{implement} this decision until you actually observe the query.

\subsubsection{Fixed-Threshold Attacks}

We begin in the authenticated setting by again considering the  fixed-threshold attacks. 
For this setting, our assumptions imply that $\Delta_j = B_j - A_j > 0$ for all $j$.
Consequently, if responding honestly to the query $q_t$ would violate privacy, we would \emph{always} wish to flip the response.
This is captured in the following proposition.
\begin{proposition}
\label{P:exactlyone}
In the fixed-threshold threat model and authenticated access setting, if $\exists~i \in B$ such that $L_i(Q_t,x,y) < \theta$, then $F_t = \{q_t\}$.
\end{proposition}
Proposition~\ref{P:exactlyone}  suggests a simple online heuristic  for ensuring privacy while minimizing the number of query flips: flip $j$ if and only if $q_t = j$ and adding $q_t$ violates privacy.
This is formalized in Algorithm~\ref{algo:authonline}.
Our experiments below will demonstrate that this simple heuristic is remarkably effective in practice.
We now observe that while this heuristic may not be optimal, it does guarantee privacy in this setting under the reasonable assumption that $\theta \le 0$ (otherwise, privacy is impossible, because the constraint is violated even before any queries are made).
\begin{algorithm}
\SetAlgoLined
\KwInput{A set $B$ of individuals in the Beacon, previous queries $Q$, history of which queries have been flipped previously $y$, the incoming query $q$, and threshold $\theta$.
}
\KwOutput{Decision whether or not to flip query $q$.}

  \If{$\exists i \in B: L_i(Q\cup q,x,y) < \theta$ }{
    \Return \textbf{Yes}
  }

  \Return \textbf{No}
 \caption{\textsc{Online Greedy} Algorithm}
 \label{algo:authonline}
\end{algorithm}
\begin{proposition}
\label{P:onlinegreedy}
Suppose that $\theta \le 0$.
Then \textsc{Online Greedy Algorithm} guarantees privacy against fixed-threshold attacks in the online authenticated access setting.
\end{proposition}
We defer the proof to Appendix~\ref{A:a4}.
\subsubsection{Adaptive Attacks}

As before, adaptive attacks complicate things considerably, but we can nevertheless leverage the algorithmic idea developed for fixed-threshold attacks.
In the case of adaptive attacks, recall that the privacy condition for each $i \in B$ becomes
\[
\sum_{j \in Q_{t,1}} \Delta_{ij}^{(K)}y_j + \eta_i(Q_t) \ge \eta^{(K)}(Q_t),
\]
where we now make the dependence of $\eta_i(Q_t)$ and $\eta^{(K)}(Q_t)$ explicit.
Note that we can still apply the \textsc{Online Greedy Algorithm} above, but with an important change: now, both $\eta_i(Q_t)$ and $\eta^{(K)}(Q_t)$ must be updated after receiving each query $q$.
Modulo this change, the algorithm, upon observing a query $q$, checks whether $\exists i \in B$ such that $\sum_{j \in Q_{t,1}} \Delta_{ij}^{(K)}y_j + \eta_i(Q_t) < \eta^{(K)}(Q_t)$, flips $q$ if this is true, and does not otherwise.

The crucial issue, however, is that we can no longer guarantee privacy in this setting, since flipping a query $q$ may now actually cause the privacy condition for some other individual to be violated.
However, in our experiments below we show that our populations remain well-mixed in terms of LRT statistics (for which the \emph{adaptive} privacy condition is a proxy).

\subsection{Unauthenticated Access}
The key distinction between authenticated and unauthenticated access in our model is that in the latter case the Beacon \emph{does not know which queries have previously been made} when it receives a new query $q$ at any given point in time.
We therefore model this setting by assuming that the query sequence (besides $q$) is adversarial.
Specifically, the privacy constraint now takes the following form:
\begin{align}
\label{E:advprivacy}
\forall i,t, \quad \min_{Q_i \subseteq Q_{t-1}} L_i(Q_i \cup q_t,x,y_t) - \theta(Q_i \cup q_t) \ge 0.
\end{align}
We use notation $Q_i$ to emphasize that since we do not know the past query sequence and wish to protect the privacy of \emph{every} $i \in B$, we are assuming that the sequence of queries is independently adversarial for each $i$.
We now show that in the unauthenticated setting, the temporal aspect collapses, and the optimal decision about which queries to flip can be made at time $t=0$.
\begin{proposition}
\label{P:upfront}
In the unauthenticated access setting, if all $j \in S$ are queried by some finite time $t$, there exists an optimal solution to the \textsc{Beacon Privacy Problem}, then there is an optimal solution with the property that $F = F_0$ and $F_t = \emptyset$ for all $t > 0$.
\end{proposition}
\begin{proof}[Proof Sketch]
For a sufficiently large $t$, $Q_{t-1} = S$. As such, $Q_i \cup q_t = Q_i$.  Then, it must be true that $F_t = \emptyset$ and $F = \cup_{t' < t} F_t$.
Since $F$ must guarantee privacy for all queries at time $t' \ge t$, it must be a minimal set of queries to do so, and we can simply identify such a set at $t=0$.
\end{proof}
This means that the unauthenticated online access setting is effectively a worst-case batch setting, where the worst case set of queries is chosen independently for each $i \in B$.
Note that this proposition appears to contradict Proposition~\ref{P:oneflip}, but in fact it does not, as neither claims that the optimal solution it characterizes is \emph{unique}.
In this case, too, we can wait to implement the flips in $F$ until the associated queries are actually observed for the first time.

The consequence of Proposition~\ref{P:upfront} is that we can simplify somewhat the definition of privacy in the unauthenticated setting:
\begin{align}
\label{E:advprivacysimple}
\forall i, \quad \min_{Q_i \subseteq S} L_i(Q_i,x,y) -\theta(Q_i) \ge 0.
\end{align}

\subsubsection{Fixed-Threshold Attacks}

Recall that $P_i(Q) = \{j \in Q_1|d_{ij} = 1\}$.
While we previously omitted the dependence of $P_i$ on $Q$, this must be explicit in the online setting.
The next proposition shows that in the case of fixed-threshold attacks, the privacy condition reduces to a particularly simple form.
\begin{proposition}
\label{P:advftatt}
In the unauthenticated access setting with fixed-threshold attacks, the privacy condition~\eqref{E:advprivacysimple} is equivalent to
\begin{align}
\label{E:advprivft}
\forall i, \quad  \sum_{j \in P_i(S) \setminus F} A_j \ge \theta.
\end{align}
\end{proposition}
The proof is deferred to Appendix~\ref{A:a5}. An important  implication of Proposition~\ref{P:advftatt} is  that, in this setting, \emph{flipping queries is equivalent to masking them}.
The reason is that since flipping increases LRT statistics, the worst-case subset of queries will never include any queries that have been flipped, effectively masking all of them.

As a consequence of Proposition~\ref{P:advftatt}, we can represent the solution to the \textsc{Beacon-Privacy-Problem} in this setting as the following integer linear program:
\begin{align}
\min_{y} \quad \sum_j y_j \quad \mathrm{subject\ to:}\quad \sum_{j \in P_i} |A_j|y_j \ge \theta - \sum_{j \in P_i} A_j,
\end{align}
where $|A_j|$ refers to the absolute value of $A_j$.
Moreover, in the special case that AAFs follow the beta distribution and we use their expectations, we can make direct use of the methods from Section~\ref{AAF_Beta}, including the \textsc{Greedy $k$-Cover} algorithm (with $k_i = \theta + |P_i|$), where $|P_j|$ is the size of the set $P_j$ (slightly overloading notation).
Similarly, even in the general case, we can leverage the heuristic algorithm in Section~\ref{S:ftgen}, replacing $\Delta_j$ with $|A_j|$.
Finally, we observe that in the special case $\theta = 0$, there is only one feasible solution, which is $F = \cup_i P_i$.
On the other hand, this solution is always feasible (but not necessarily optimal) if $\theta \le 0$, and we can thus always guarantee privacy in such a setting for fixed-threshold attacks.

\subsubsection{Adaptive Attacks}

Recall that, even in the batch setting, since $\Delta_{ij}^{(K)}$ may be negative for some $i,j$, privacy constraints may be violated for some individuals if we flip certain queries $j$.
Since in the unauthenticated setting we are making decisions up-front, we only consider the subset of queries $j$ for which $\Delta_{ij}^{(K)} \ge 0$ for all $i$.

Unpacking the condition in Equation~\eqref{E:advprivacysimple} and rearranging terms, we obtain the following privacy condition for adaptive attacks for each individual $i \in B$:
\[
\min_{Q_{i,1} \subseteq S_1} \sum_{j \in Q_{i,1}} \left(\Delta_{ij}^{(K)}y_j + d_{ij}^{(K)}A_j\right) + \min_{Q_{i,0} \subseteq S_0} d_{ij}^{(K)}B_j \ge 0,
\]
where $d_{ij}^{(K)} = d_{ij} - \sum_{k \in \bar{B}^{(K)}} \frac{d_{kj}}{K}$.
Since the second term on the left hand side doesn't depend on $y$ (equivalently, $F$), we can pre-compute it, setting $k_i = -\min_{Q_{i,0}}d_{ij}^{(K)}B_j$.
Consequently, we obtain the condition
\[
\min_{Q_{i,1}} \sum_{j \in Q_{i,1}} \left(\Delta_{ij}^{(K)}y_j + d_{ij}^{(K)}A_j\right) \ge k_i.
\]
We now use this expression to obtain a variant of the \textsc{MI Greedy} heuristic for this setting.
The key idea behind this heuristic was to choose a query $j$ to flip that has the highest average marginal impact in each iteration (omitting individuals previously "covered" in the sense that their privacy is satisfied).
Since we only consider flipping queries with $\Delta_{ij}^{(K)} \ge 0$, this will not have a detrimental impact on any such "covered" individuals, as it can only increase their LRT statistics.
For any $i$ not yet covered,
define $\mu_{ij}$ to be the marginal impact of flipping a query $j$.
If $\Delta_{ij}^{(K)} + d_{ij}^{(K)}A_j \ge 0$, this query will be omitted as a result of the flip, and the marginal contribution is thus $\mu_{ij} = |d_{ij}^{(K)}A_j|$.
If, on the other hand, $\Delta_{ij}^{(K)} + d_{ij}^{(K)}A_j < 0$, this query will remain, but its contribution reduced by $\Delta_{ij}^{(K)}$ and the marginal contribution is therefore $\mu_{ij} = \Delta_{ij}^{(K)}$ as in the batch setting.

\section{Experiments}
\subsection{Experiment Setup}

\smallskip
\noindent{\bf Dataset }
The dataset used in this paper was originally made available by the organizers of the 2016 iDash Privacy and Security Workshop~\cite{idash16} as part of their \emph{Practical Protection of Genomic Data Sharing through Beacon Services} challenge. The goal of the challenge was for teams to develop approaches that release as many truthful responses as possible through a modified beacon before the Shringarpure-Bustamante attack \cite{shringarpure2015privacy} could be used to re-identify an individual. 
In this study, we use SNVs from Chromosome 10 for a subset of 400 individuals to construct the beacon, and another 400 individuals excluded from the beacon.
Unless otherwise specified, we set the genomic sequencing error rate to $\delta = 10^{-6}$, as in the iDash challenge.

\smallskip
\noindent{\bf Computational Environment }
Experiments were carried out on a PC with an AMD Ryzen 7 3800x processor and 64 GB DDR4 3600 MHz-CL19 RAM running Ubuntu version 18.04.5, using Python version 3.6.12.

\smallskip
\noindent{\bf Baselines }
We compare our approaches to three state of the art baselines.
The first is \emph{strategic flipping (SF)}~\cite{wan2017controlling}, which is the winning entry to the 2016 iDash Privacy Challenge, and uses a combination of greedy and local search. 
We compare two versions of this approach: first, the version as previous implemented (\emph{SF}), and second, a variant, \emph{SF-M} in the adaptive settings which uses our definition of privacy instead of setting a static threshold using a maximum allowable false positive rate.
Note that \emph{SF} and \emph{SF-M} are equivalent in the fixed-threshold setting.
The second baseline is \emph{random flipping (RF)} proposed by Raisaro et al.~\cite{raisaro2017addressing}.
RF randomly flips a subset of unique alleles in the Beacon dataset by sampling from a binomial distribution.
The third baseline we consider is \emph{differential privacy (DP)} as proposed for this setting by Cho et al.~\cite{cho2020privacy}.
These baselines are configured so that in the fixed-threshold batch setting they maximize utility within their respective parameter configuration space while guaranteeing privacy.

\subsection{The Batch Setting}
\subsubsection{Fixed-Threshold Attacks with a Small Sequencing Error}

In these experiments, we assume that genomic sequencing error $\delta$ is extremely small (we set it to $\delta=10^{-240}$ in this set of experiments), and it suffices to flip a single beacon response per individual to guarantee privacy in this setting. 
Because flipping queries when the minor allele is very frequent is likely to degrade the trust in the system, we consider the effect of a restriction on the rarity of occurrence of the alternate allele on the number of flips needed to secure privacy.

\begin{figure}[h]
\centering
\captionsetup[subfigure]{justification=centering}
	\begin{subfigure}[t]{0.5\linewidth}
	    \centering
	    \includegraphics[width=\textwidth]{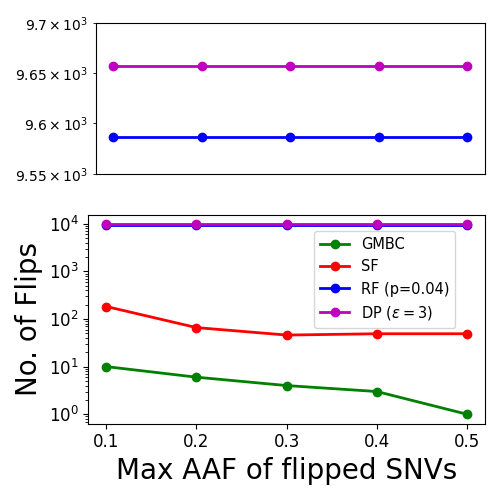}
	    \caption{Small $\delta$ setting.}
	    \label{fig:simple_a}
	\end{subfigure}%
	\hfill
	\begin{subfigure}[t]{0.5\linewidth}
	    \centering
	    \includegraphics[width=\textwidth]{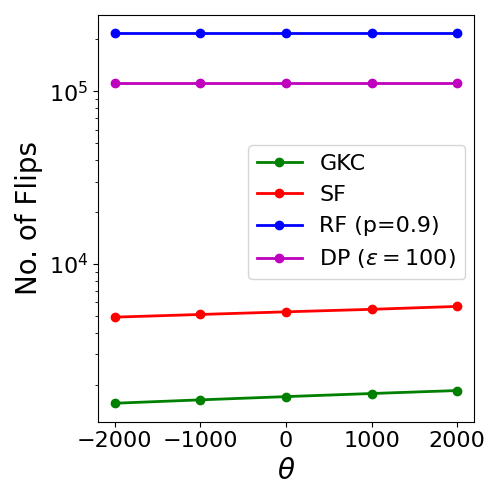}
	    \caption{AAFs drawn from a Beta distribution.}
	    \label{fig:simple_b}
	\end{subfigure}
	\caption{Number of SNVs flipped to guarantee privacy over 400 individuals in a beacon compared to baselines. 
	}
	\label{fig:simple_cases}
\end{figure}

Fig.~\ref{fig:simple_a} compares the number of flipped queries between the proposed Greedy \textsc{Min-Beacon-Cover} (GMBC) and the three baselines.
We can see that GMBC allows us to guarantee privacy with significantly (more than an order of magnitude) fewer false beacon responses compared to the baselines. The suboptimality of strategic flipping stems from not accounting for how many individuals a SNV affects, and only looking at the average over the population, a limitation that GMBC overcomes. 

\subsubsection{Fixed-Threshold Attacks with AAFs Drawn from Beta Distribution}
Next, we consider the setting with a static prediction threshold $\theta$, where the alternate allele frequencies (AAFs) are assumed to be drawn from a Beta-distribution. Recall that in this setting, it suffices to flip a minimum number of SNVs, $k_i$, per individual. Once again, we present results comparing to the three baselines, now varying the value of $\theta$. From here on, we limit ourselves to experiments where there are no restrictions on how frequently an alternate allele can be present in a population, due to the much larger compute times needed to handle multiple high-precision values arising from the AAFs for over 1.3 million SNVs. 
Henceforth, we also set $\delta = 10^{-6}$.

Fig.~\ref{fig:simple_b} presents the results comparing the proposed \textsc{GKC} approach to the baselines for $\theta \in [-2000-2000]$.
Again, we see that the proposed GKC algorithm has significantly higher utility (fewer flips) compared to the alternatives.

\subsubsection{General Case - True AAFs}
\begin{figure}[h]
  \centering
  \captionsetup[subfigure]{justification=centering}
  \begin{subfigure}[t]{0.47\linewidth}
      \centering
      \includegraphics[width=\textwidth]{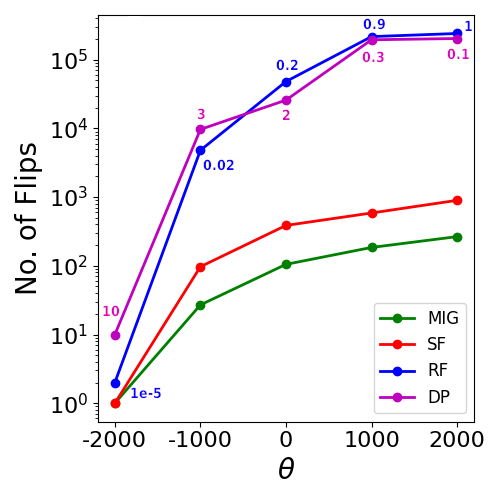}
      \caption{Utility comparison in the general AAF setting.}
      \label{fig:general_a}
  \end{subfigure}
  \begin{subfigure}[t]{0.47\linewidth}
      \centering
      \includegraphics[width=\textwidth]{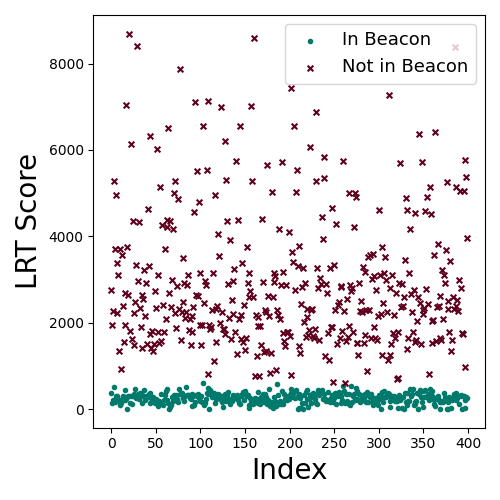}
      \caption{LRT scores of individuals in and not in the Beacon.}
      \label{fig:general_b}
  \end{subfigure}
    \caption{Performance in the fixed threshold batch setting with general AAFs. 
    }
  \label{fig:general_AAF}
\end{figure}
Next, we look at the more general case with no assumptions on AAFs, a realistic $\delta = 10^{-6}$, and an adversary who computes a static threshold $\theta$ based on some prior knowledge. 
This setting is much more representative of a real-world attack. 
Fig.~\ref{fig:general_a} compares the number of queries flipped by \textsc{Marginal-Impact Greedy (MIG)} to the baselines, over a range of prediction thresholds $\theta$. In this figure, also note that for the differential privacy and random flipping baselines, for each value of $\theta$, we present the performance with an empirically selected parameter ($\epsilon$ and $p$ respectively) which yields the highest utility while preserving privacy for all individuals, with the corresponding parameter denoted in the plot.
We can observe that \textsc{MIG} again outperforms both \textsc{RF} and \textsc{DP} by several orders of magnitude in terms of utility (all approaches preserve privacy of all individuals in the beacon dataset).
\textsc{SF} is closer to \textsc{MIG}, but still flips considerably more queries.

\subsubsection{Adaptive Attacks}
While the algorithms devised with a fixed $\theta$ in mind guarantee privacy given this assumption, adaptive attackers can defeat these approaches by taking the revised Beacon queries explicitly into account when determining the threshold.
This is illustrated in Fig.~\ref{fig:general_b}, which shows the LRT scores for 400 individuals in the beacon, and 400 others not in the beacon after flipping responses using \textsc{MI-Greedy} for $\theta=0$. 
While the LRT scores for the individuals in the Beacon do end up above 0, they also remain below the scores computed on those not in the Beacon.
Since the attacker does not in fact know who is in the Beacon, a simple clustering attack can indeed separate the two populations.
Specifically, using $1$-dimensional $k$-means clustering of the LRT scores achieves $100\%$ true positive rate at the cost of a $30\%$ false positive rate, averaged over 20 runs. Next, we evaluate the effectiveness of the proposed algorithms that aim to explicitly account for this more sophisticated attack.
\begin{figure}[h]
  \centering
  \captionsetup[subfigure]{justification=centering}
  \begin{subfigure}[t]{0.47\linewidth}
      \centering
      \includegraphics[width=\textwidth]{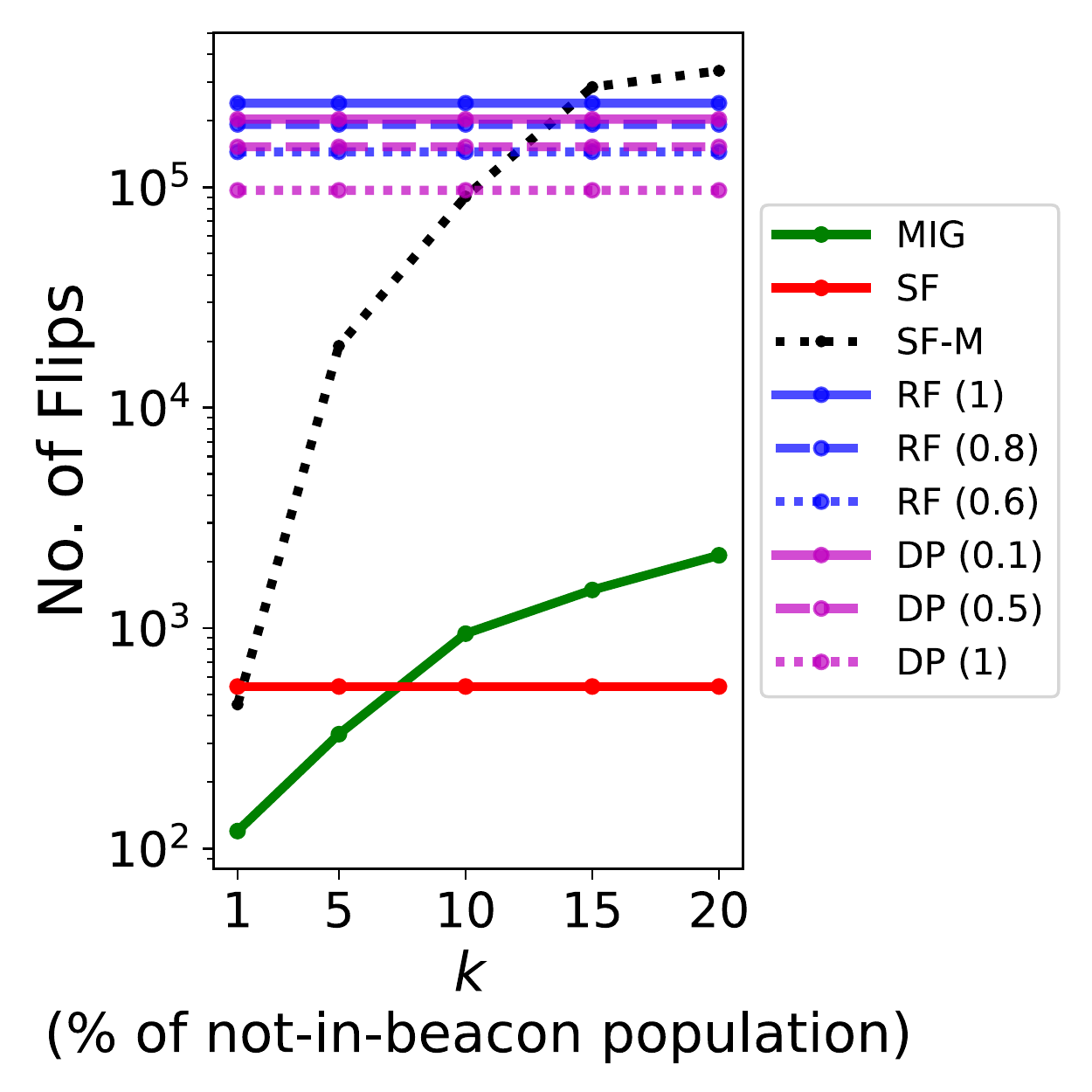}
      \caption{Utility comparison.}
      \label{fig:adaptive_a}
  \end{subfigure}
  \begin{subfigure}[t]{0.47\linewidth}
      \centering
      \includegraphics[width=\textwidth]{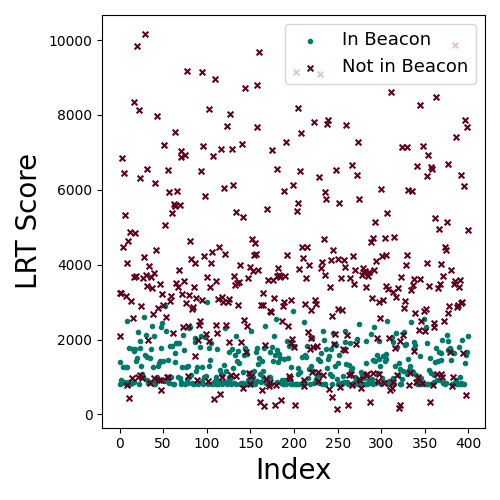}
      \caption{LRT scores after using the adaptive attack variant of \textsc{MIG}.}
      \label{fig:adaptive_b}
  \end{subfigure}
    \caption{Performance in the adaptive attack batch setting.}
  \label{fig:adaptive}
\end{figure}

Fig.~\ref{fig:adaptive_a} shows the number of flips that need to be flipped by our \emph{adaptive attack} variant of \emph{MIG}, as well as by the various baselines in this setting. The values of $\epsilon$ and $p$ used for DP and RF respectively are shown in parentheses in the plot legend. Appendix \ref{A:a6} presents a comparison of privacy achieved and the ROC curves. While MIG and SF-M preserve privacy of all individuals in the Beacon in this setting, our approach does so while flipping orders of magnitude fewer query responses. Both DP (for $\epsilon \in \{0.1, 0.5, 1\}$) and SF fail to achieve privacy for all individuals.
Fig.~\ref{fig:adaptive_b} illustrates that explicitly accounting for adaptive attacks, \textsc{MIG} yields much more mixed LRT scores between individuals in and not in the Beacon dataset. 
Quantitatively, setting $K=20$ (see Section~\ref{S:adaptive}), the false positive rate for the clustering attack increases from $30\%$ to over $50\%$ on average.

\subsection{The Online Setting}
\subsubsection{Authenticated Access} Recall that in the authenticated setting, the defender has access to each user's query history, and thus a decision about whether to flip the beacon response for a SNV can be greedily made at runtime. 
\begin{figure}[h]
  \centering
  \captionsetup[subfigure]{justification=centering}
  \begin{subfigure}[t]{0.47\linewidth}
      \centering
      \includegraphics[width=\textwidth]{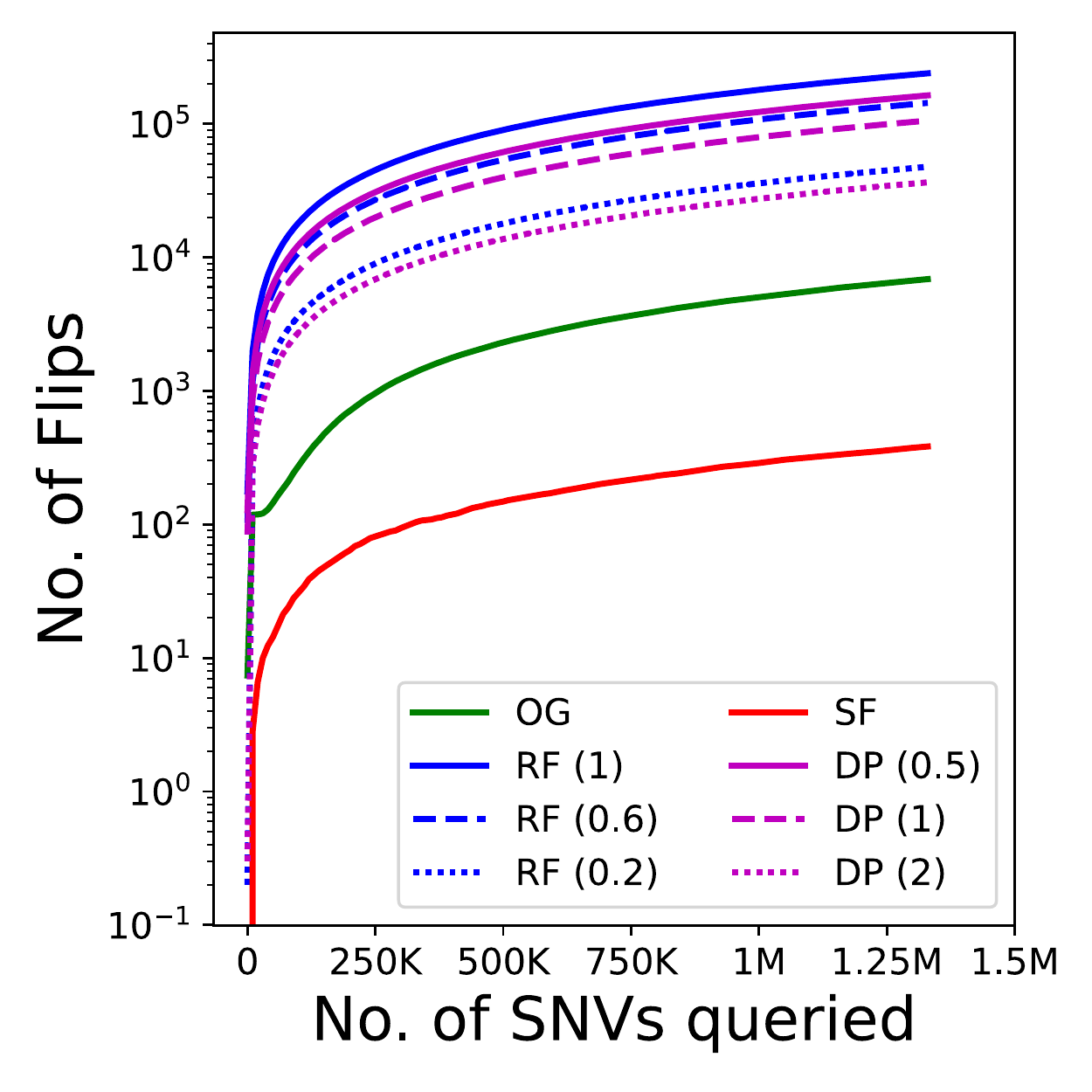}
      \caption{Utility comparison.}
      \label{fig:authOnline_fixed_a}
  \end{subfigure}
  \begin{subfigure}[t]{0.47\linewidth}
      \centering
      \includegraphics[width=\textwidth]{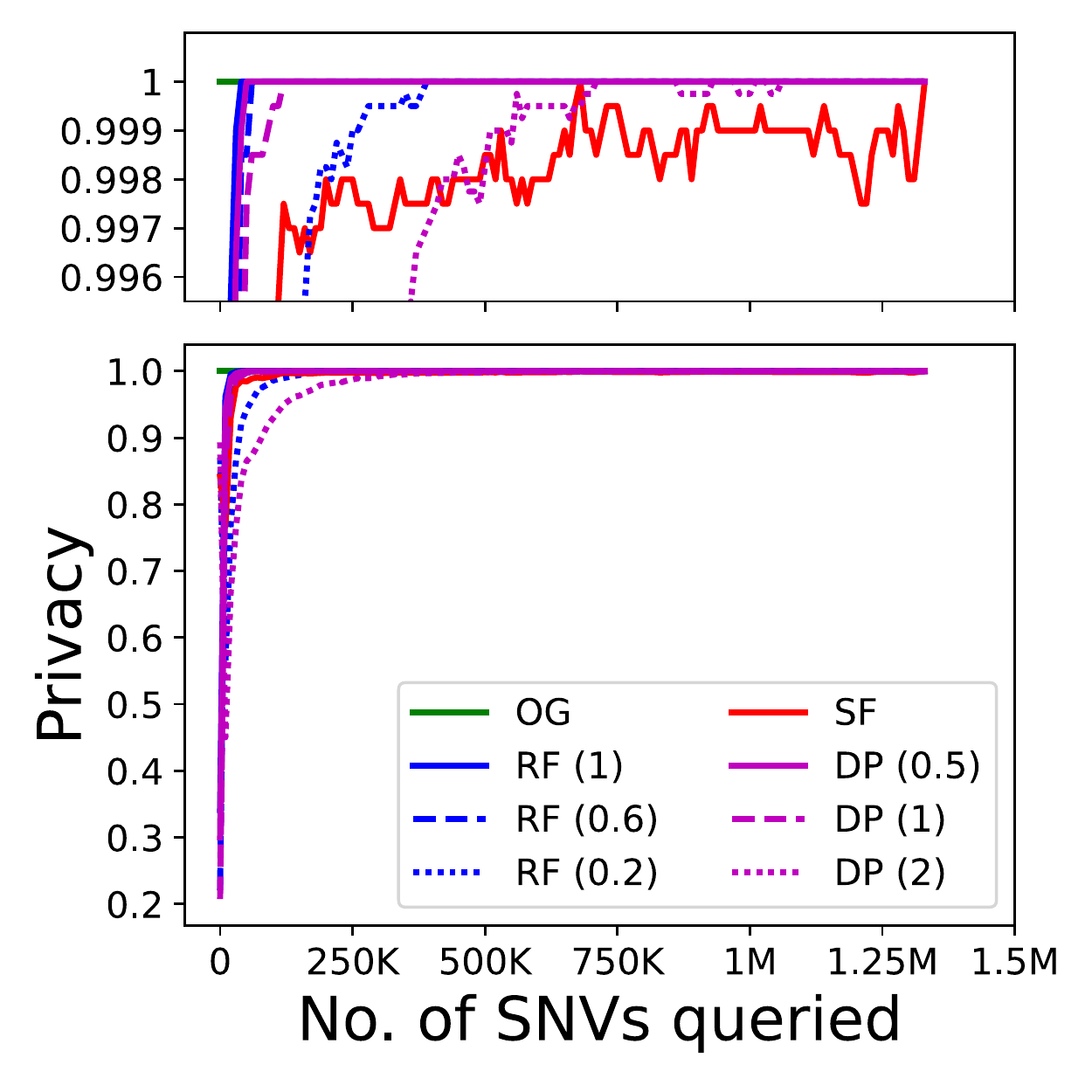}
      \caption{Privacy comparison.}
      \label{fig:authOnline_fixed_b}
  \end{subfigure}
    \caption{Performance in the authenticated online setting with fixed-threshold attacks; $\theta = 0$.}
  \label{fig:authOnline_fixed}
\end{figure}
Unlike previous settings where all methods were able to achieve perfect privacy for all individuals in the beacon, this will no longer always be the case in the online setting. Consequently, we also compare our methods with the baselines in terms of privacy, defined as the fraction of the individuals in the Beacon whose privacy is not violated.

\begin{figure}[h]
  \centering
  \captionsetup[subfigure]{justification=centering}
  \begin{subfigure}[t]{0.47\linewidth}
      \centering
      \includegraphics[width=\textwidth]{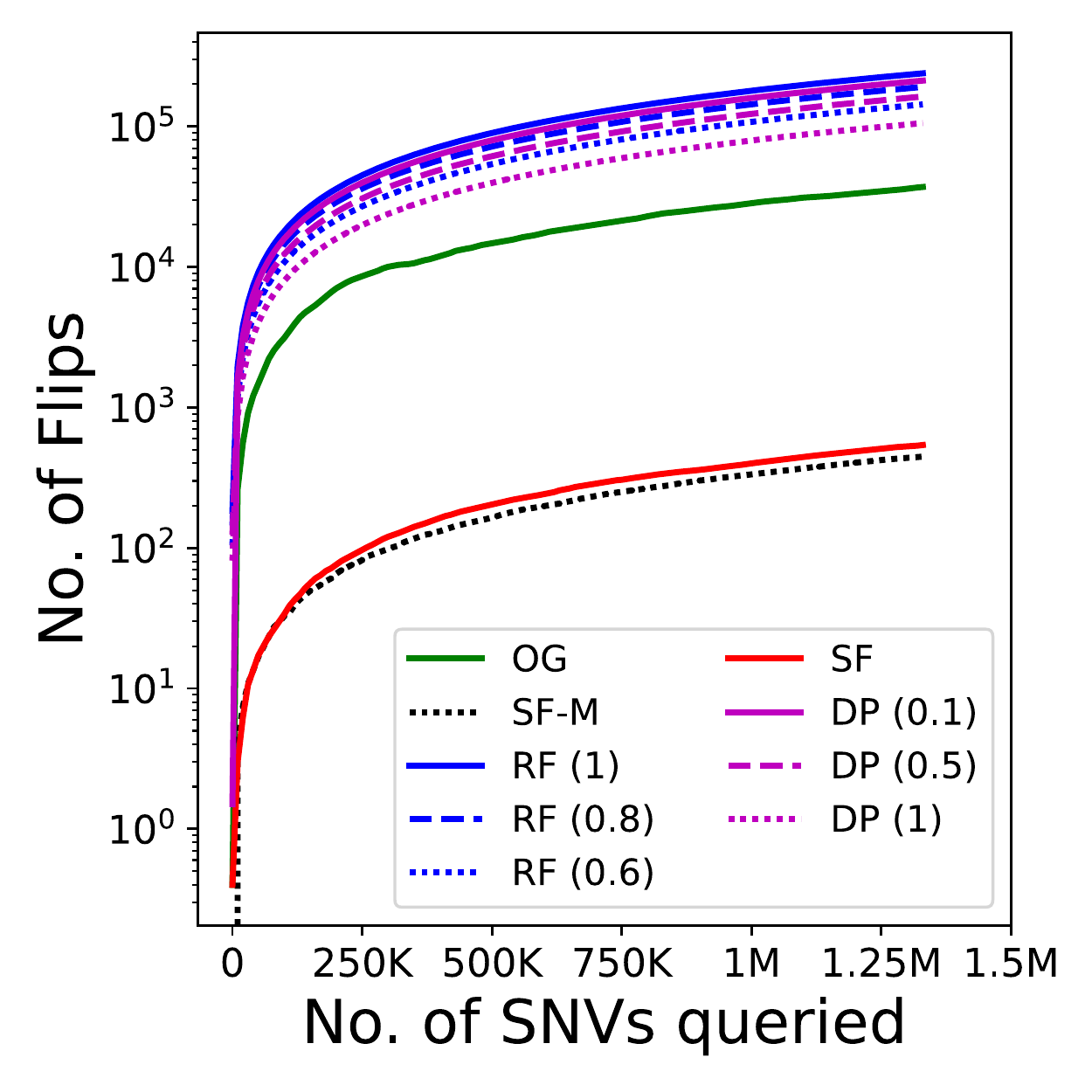}
      \caption{Utility ($K=1$).}
      \label{fig:authOnline_ada_a}
  \end{subfigure}
  \begin{subfigure}[t]{0.47\linewidth}
      \centering
      \includegraphics[width=\textwidth]{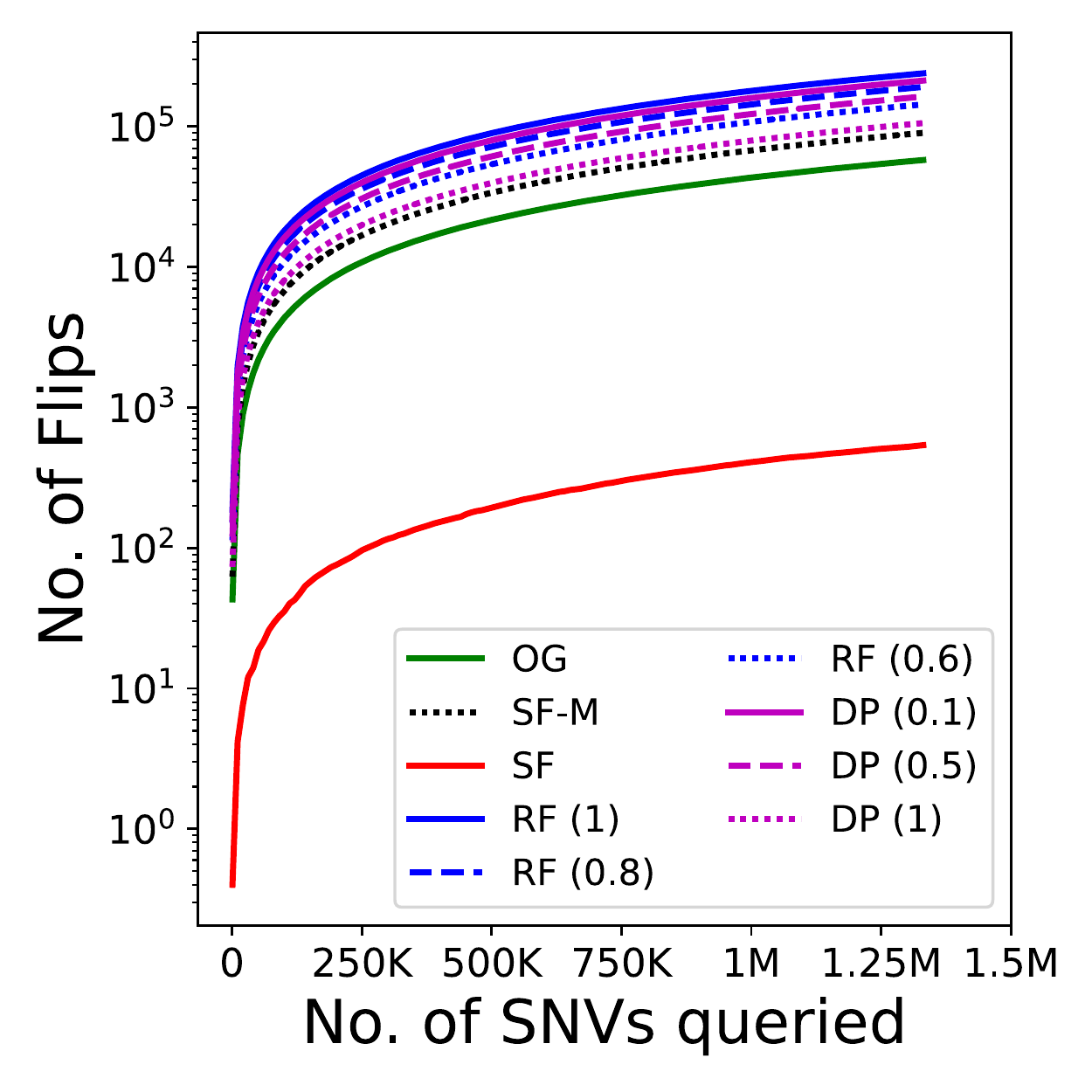}
      \caption{Utility ($K=10$).}
      \label{fig:authOnline_ada_b}
  \end{subfigure}
  \newline
  \begin{subfigure}[t]{0.47\linewidth}
      \centering
      \includegraphics[width=\textwidth]{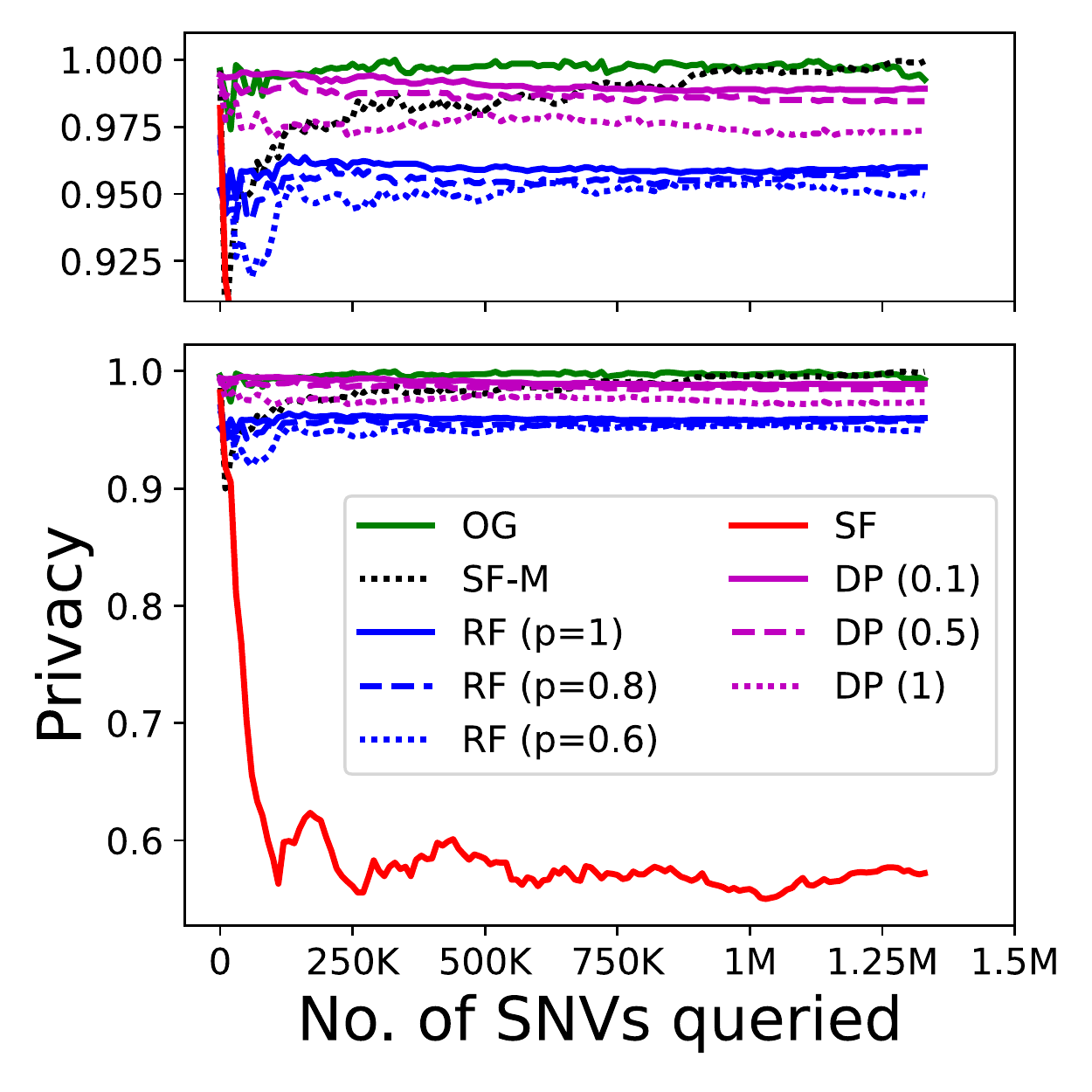}
      \caption{Privacy ($K=1$).}
      \label{fig:authOnline_ada_c}
  \end{subfigure}
  \begin{subfigure}[t]{0.47\linewidth}
      \centering
      \includegraphics[width=\textwidth]{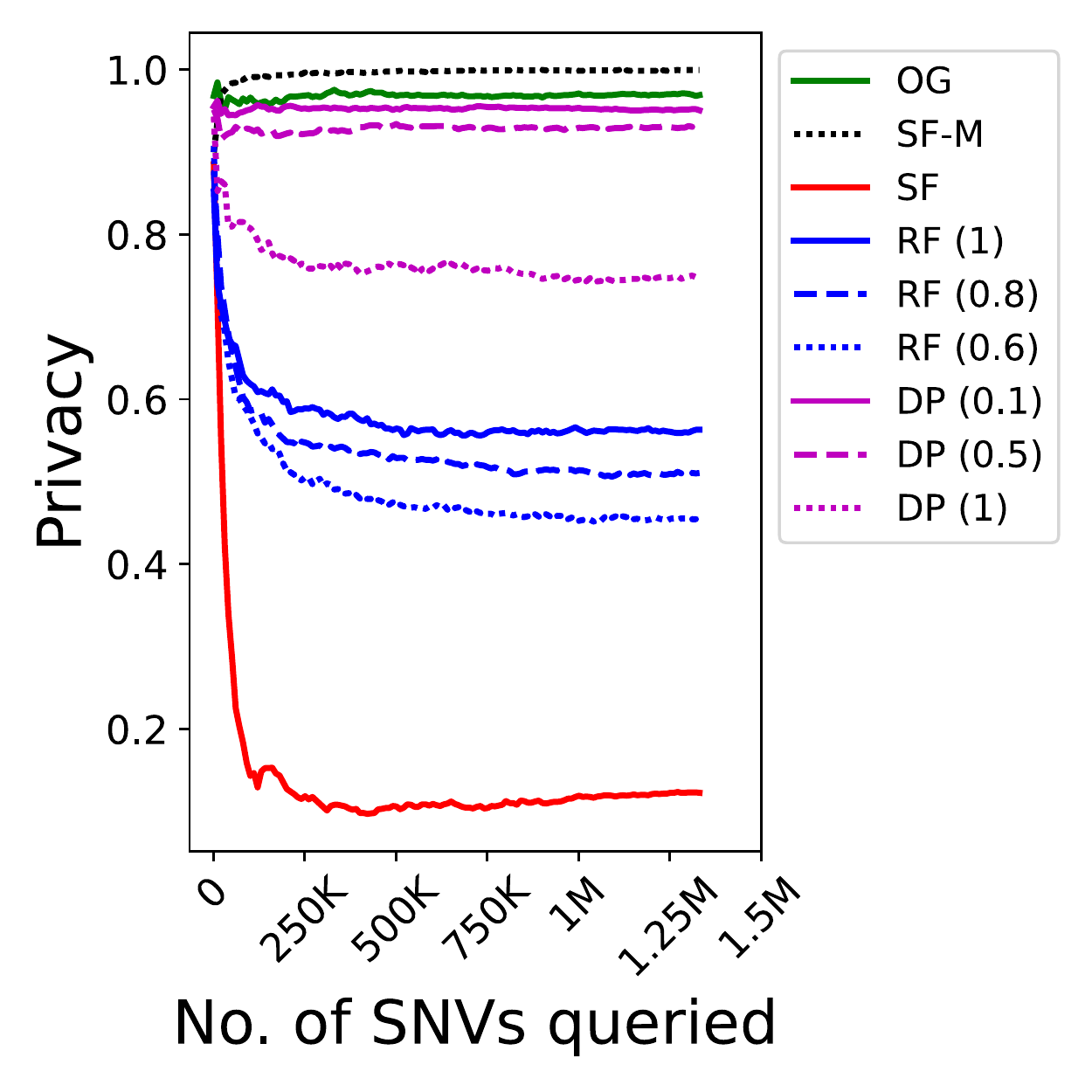}
      \caption{Privacy ($K=10$).}
      \label{fig:authOnline_ada_d}
  \end{subfigure}
    \caption{Performance comparison in the authenticated online setting with adaptive attacks.}
  \label{fig:authOnline_ada}
\end{figure}
First, we consider fixed-threshold attacks.
Fig.~\ref{fig:authOnline_fixed} compares \textsc{Online-Greedy} (OG) with the baselines when $\theta=0$ (the worst case threshold in the online setting).
Recall that OG provably achieves privacy in such settings, but does flip more queries than \textsc{SF} (but far fewer than other baselines).
In contrast, \textsc{SF} does compromise privacy of a considerable fraction of individuals in the Beacon (as do other baselines) when few SNVs have been queried.

In Fig.~\ref{fig:authOnline_ada}, we consider adaptive attacks. When $K=1$, \textsc{OG} tends to have better privacy, but considerably lower utility than \textsc{SF} and \textsc{SF-M} (and dominates the other two baselines in both).
For $K=10$, however, \textsc{OG} has better utility that all baselines but \textsc{SF}, but slightly lower privacy than \textsc{SF-M} (and better than others).
While \textsc{SF} achieves the highest utility, it has extremely poor privacy.
The corresponding ROC curves are provided in Appendix~\ref{A:a7}.

\subsubsection{Unauthenticated Access} 
Finally, we compare the proposed approaches to baselines in the unauthenticated online setting.
Once again, we begin with fixed-threshold attacks.
As shown in Fig.~\ref{fig:unauthOnline_fixed}, \textsc{OMIG} (our algorithm variant for this setting) flips more queries than \textsc{SF}, but does guarantee privacy, whereas \textsc{SF} compromises the privacy of a subset of individuals, particularly as more SNVs can be queried.
The other two baselines also achieve privacy, but at a considerable loss in utility compared to \textsc{OMIG}.

\begin{figure}[h]
  \centering
  \captionsetup[subfigure]{justification=centering}
  \begin{subfigure}[t]{0.47\linewidth}
      \centering
      \includegraphics[width=\textwidth]{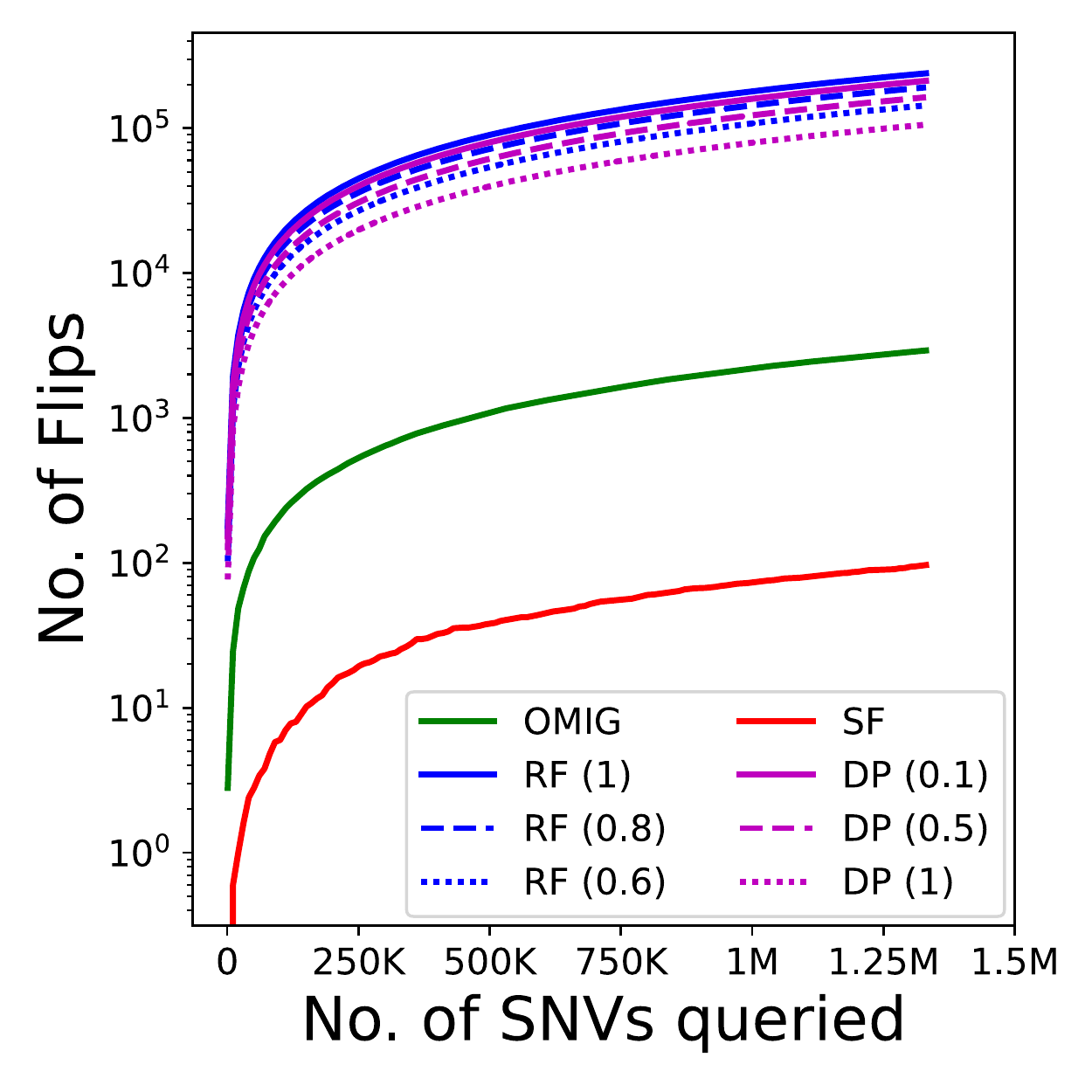}
      \caption{Utility comparison.}
      \label{fig:unauthOnline_fixed_a}
  \end{subfigure}
  \begin{subfigure}[t]{0.47\linewidth}
      \centering
      \includegraphics[width=\textwidth]{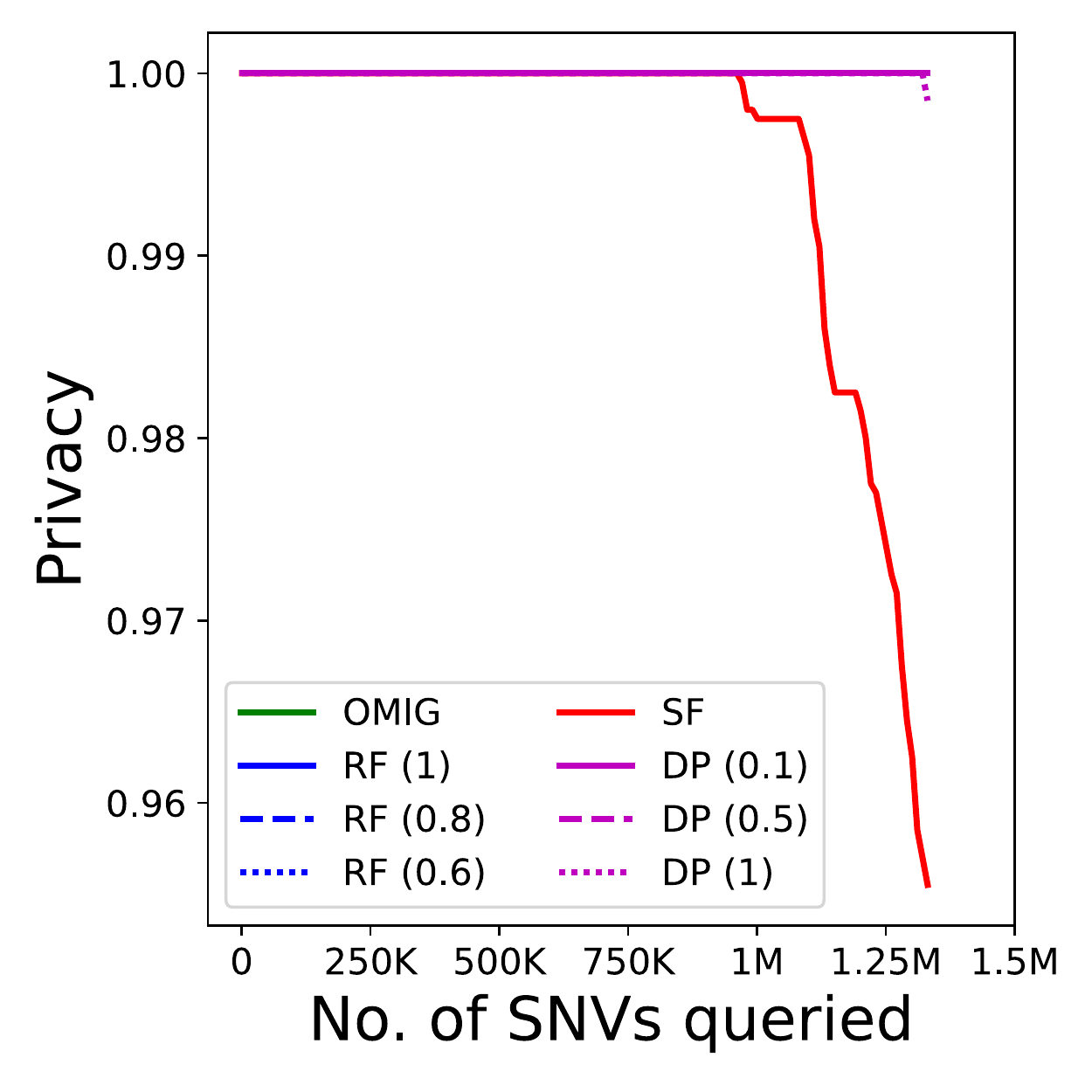}
      \caption{Privacy comparison.}
      \label{fig:unauthOnline_fixed_b}
  \end{subfigure}
  \caption{Performance in the online unauthenticated fixed-threshold attack setting; $\theta = -1000$.}
  \label{fig:unauthOnline_fixed}
 \end{figure}
  \begin{figure}[h]
  \centering
  \captionsetup[subfigure]{justification=centering}
  \begin{subfigure}[t]{0.47\linewidth}
      \centering
      \includegraphics[width=\textwidth]{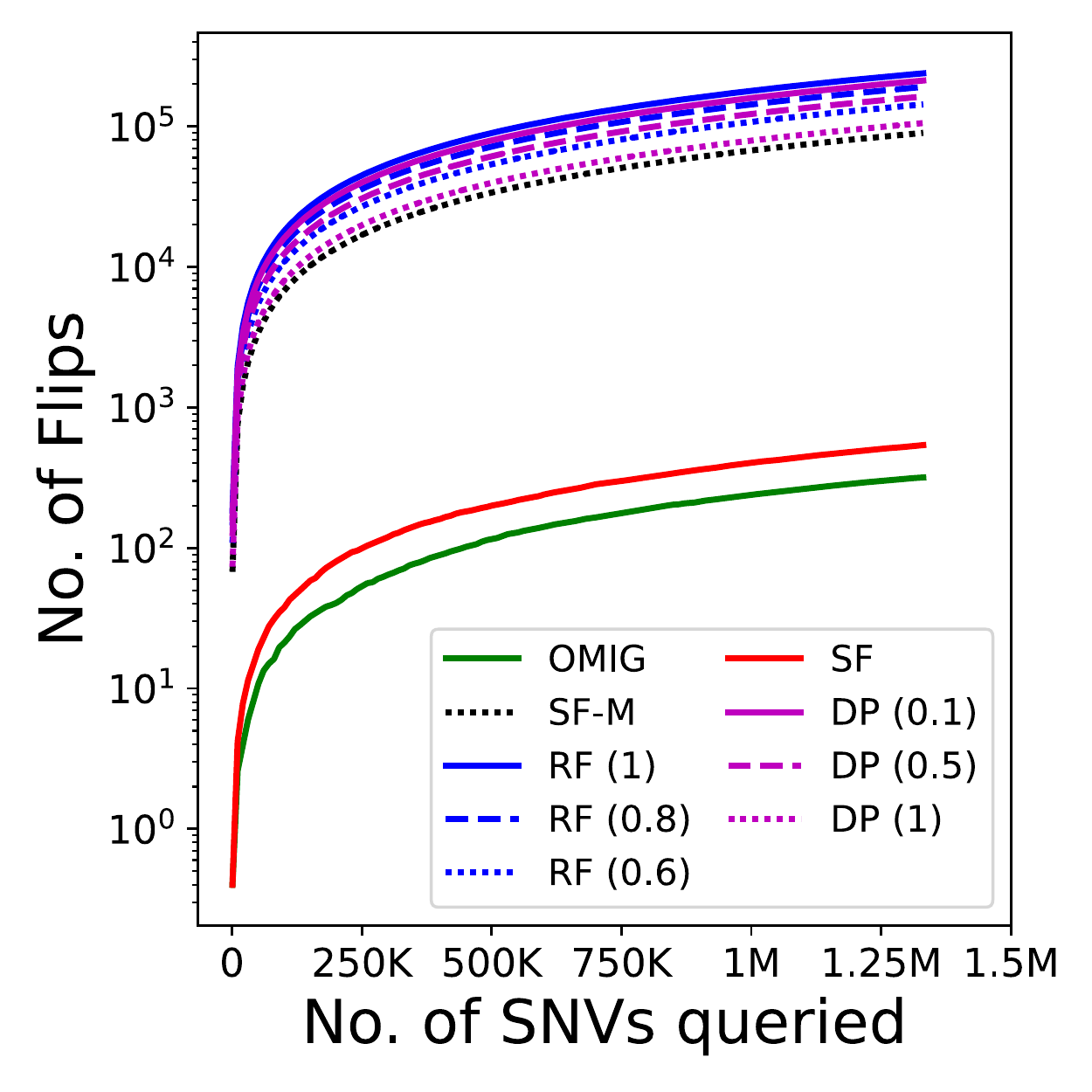}
      \caption{Utility comparison.}
      \label{fig:unauthOnline_ada_a}
  \end{subfigure}
  \begin{subfigure}[t]{0.47\linewidth}
      \centering
      \includegraphics[width=\textwidth]{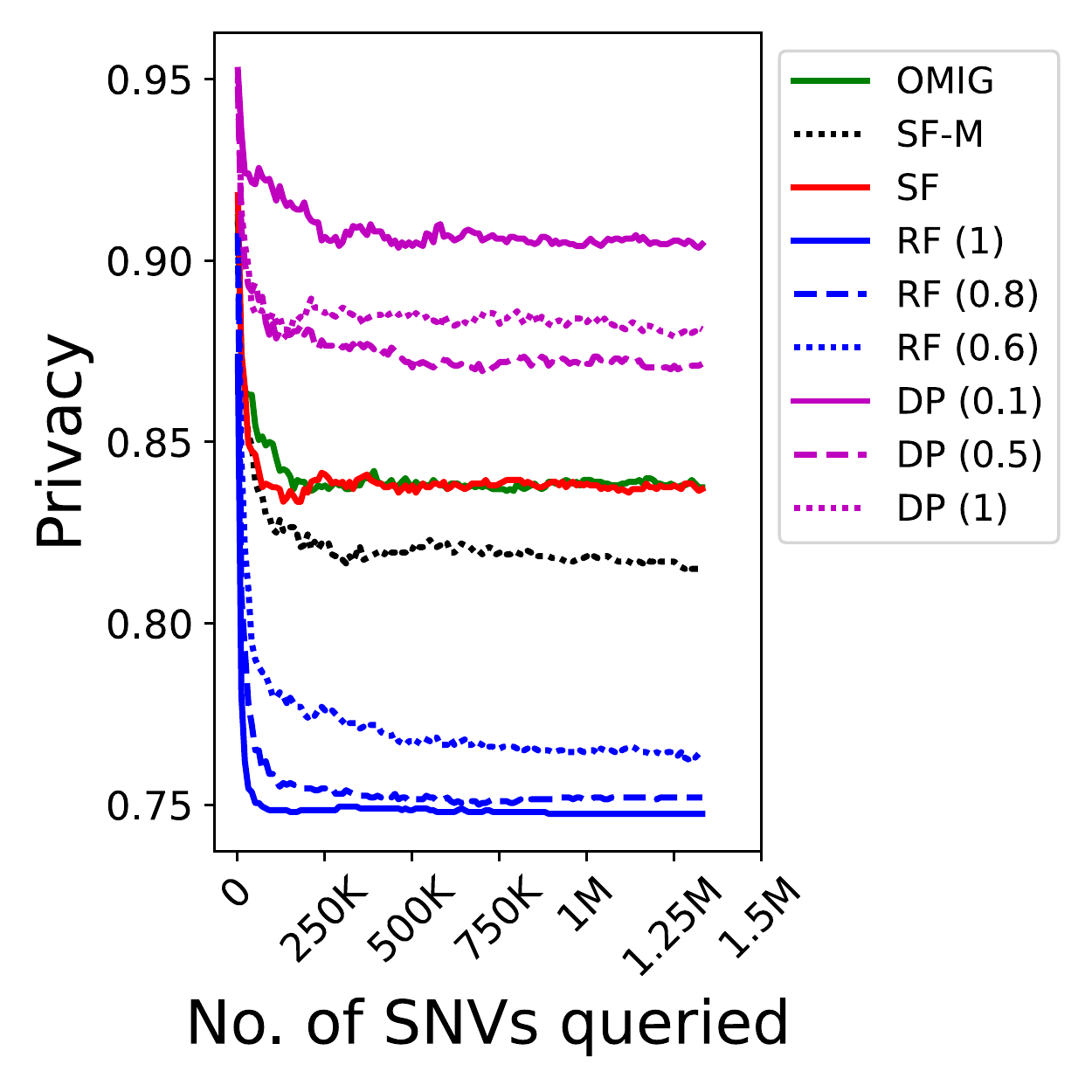}
      \caption{Privacy comparison.}
      \label{fig:unauthOnline_ada_b}
  \end{subfigure}
  \caption{Performance in the online unauthenticated adaptive attack setting; $K=10$.}
  \label{fig:unauthOnline_ada}
 \end{figure}
 
 Fig.~\ref{fig:unauthOnline_ada} presents a similar comparison for adaptive attacks.
 In this setting, all methods including \textsc{OMIG} now compromise privacy, with \textsc{DP} doing so the least.
 However, \textsc{OMIG} is now again orders of magnitude better than most of the baselines in terms of utility (with \textsc{SF-M} now performing relatively poorly in terms of both utility and privacy). While \textsc{SF} performs similarly to OMIG in terms of privacy, it has slightly lower utility.ROC curves are provided in Appendix~\ref{A:a8}.
 
\section{Related Work}

\noindent{\bf Privacy Violation of Shared Genomic Data: }
Genomic data sharing picked up pace with programs such as the Database of Genotypes and Phenotypes (dbGaP) by the NIH as a centralized repository of large-scale genomic data for researchers which made summary statistics about genomic data publicly available. However, shortly thereafter, it was shown by Homer et al. \cite{homer2008resolving} that it was indeed possible to resolve individuals in a DNA mixture - even if the individual's contribution to the mixture is as little as 1\% by comparing allele frequencies obtained from probe-intensities, to the allele frequencies of a reference population such as from the International HapMap Project \cite{gibbs2003international}. Various types of attacks on genomic data are addressed in \cite{erlich2014routes}, spanning identity tracing, attribute-disclosure and completion attacks. Since the Shringarpure-Bustamante attack in \cite{shringarpure2015privacy}, further attacks on beacon-services \cite{raisaro2017addressing, von2019re, ayoz2020genome, ayoz2020effect} have been proposed, including genomic-reconstruction and addressing the effects of kinship.

\noindent{\bf Protecting the Privacy of Genomic Data: }
In response, several defenses were proposed against attacks on genomic data - usually relying on injecting noise, or by suppressing a subset of the data. The efficacy of masking a subset of shared genomic data was studied in \cite{sankararaman2009genomic}. Preserving privacy was studied through a game-theoretic lens in \cite{wan2017expanding}. Several other genomic privacy-preserving methods are summarized in \cite{bonomi2020privacy}. Specific to beacon-services, \cite{raisaro2017addressing} present a random-flipping heuristic which perturbs unique alleles in the database using a binomial distribution, as well as a query-budget approach for authenticated users. A differential-privacy approach was adopted in \cite{cho2020privacy}, which forms one of the baselines in this paper. The winning entry \cite{wan2017controlling} to the 2016 iDash Privacy and Security Workshop challenge defines a differential discriminative power to select SNVs from the dataset for which responses are flipped.
\section{Discussion and Limitations}

We presented a novel framework for privacy-preserving design of Beacons in the context of membership inference attacks leveraging a likelihood-ratio-test statistic.
Our framework precisely dissects the many ways in which the Beacon service can be configured and used, such as allowing queries as a batch or in a sequence, and allowing authenticated access to individuals whose identities can be verified, or simply opening the service to the public.
We also consider two distinct threat models, one of which has been explicitly studied in prior literature, while the second involves a stronger adaptive attack and has not been formally defined or analyzed in prior work.
We present algorithms that exhibit privacy guarantees for some of these instantiations of our model, and in one special case, a provable approximation of optimal utility (while guaranteeing privacy).
All our algorithms run in polynomial time, and are highly scalable.
Moreover, the proposed algorithms typically outperform prior art in either privacy, utility, or both.

Our approach, however, has several limitations.
First, our privacy model is specific to the Beacon service, and the LRT-based attack; it is possible that other attacks can be devised that can defeat our approach.
Second, flipping query responses, while common in prior art, is not always a viable means to protect the Beacon service (for example, it may degrade public trust in the service). An alternative framework of \emph{masking} a subset of SNV queries may offer another practical solution without this limitation, but may in turn result in an even greater degradation of the utility of the Beacon service.
\bibliographystyle{acm}
\bibliography{sample-base}
\appendix
\section*{Appendix}

\section{Proof of Proposition~\ref{lemma_flip}}
\label{A:a1}

  Consider the $j^{th}$ query. If individual $i$ does not have an alternate allele at position $j$, flipping the beacon response makes no difference (refer Eq.~\ref{eq:LRT}; $d_{ij}=0$ when individual does not have alternate allele at position $j$). When the individual does have an alternate allele at position $j$ (i.e. $d_{ij}=1$), changing the beacon response $x_j$ from $0$ to $1$ changes the contribution of query $j$ to the LRT score from $\log \frac{D^j_n}{\delta D^j_{n-1}}$ to $\log \frac{1-D^j_n}{1-\delta D^j_{n-1}}$. Given sampling error $\delta < \frac{D^j_n}{D^j_{n-1}}$, dividing on both sides by $\delta$, we have
  \(
  \frac{D^j_n}{\delta D^j_{n-1}} > 1~ \textrm{(as $\delta\ge 0$)},\) and, consequently,
  \(
  \log \frac{D^j_n}{\delta D^j_{n-1}} > 0. 
  \)
  Since 
  \(
  \delta < \frac{D^j_n}{D^j_{n-1}},
  \)
  multiplying both sides by $D^j_{n-1}$ yields
  \(D^j_n > \delta D^j_{n-1}\) (since $D^j_{n-1} \in [0,1]$), which implies that
  \(1-D^j_n < 1-\delta D^j_{n-1}\).
  Dividing both sides by $1-\delta D^j_{n-1}$, we have
  \(\frac{1-D^j_n}{1-\delta D^j_{n-1}} < 1\), 
  since $\delta \in [0,1]$ and $D^j_{n-1}\in [0,1]$, which in turn implies that 
  \(\log\frac{1-D^j_n}{1-\delta D^j_{n-1}} < 0\).

\section{Proof of Theorem~\ref{T:hard}}
\label{A:a2}
First, note that \textsc{Beacon-Privacy-D} is in \NP, since given a set $F$ of flips, it is straightforward to verify that the privacy constraint holds for each individual $i$.

To prove that the problem is \NP-hard, we reduce from the \textsc{Set Cover} problem.
First, observe that in the case where $\delta = 0$, and by Proposition~\ref{lemma_flip}, to guarantee privacy of any individual $i \in B$, it suffices to flip a single response $x_j$ from 1 to 0 from all with $d_{ij} = 1$ (if we use the convention that division by $0$ results in $\infty$, any such flip causes $L_i(Q,x) = \infty$).

Now, let elements of $U$ correspond to individuals in the Beacon, i.e., $B = U$.
Let subsets $R_j$ correspond to queries $j$, where each element represents an individual $i$ with $d_{ij} = 1$.
Since without loss of generality we can assume that each $R_j$ is non-empty (since we can ignore any empty subsets in both \textsc{Set Cover}, and in the construction of \textsc{Beacon-Privacy-D} instance by Proposition~\ref{lemma_flip}), this also implies that the corresponding query response is $x_j = 1$, as at least one individual has $d_{ij} = 1$.
For any individual (element of $U$) $i \notin R_j$, we set $d_{ij} = 0$.
Furthermore, since $\cup_j R_j = U$, each individual has at least one $j$ with $d_{ij} = 1$.
Finally, the constant $k$ is now the constraint on the size of $F$, the subsets of queries to flip.

Suppose that we find the set $F$ of queries to flip that guarantees privacy.
Let $T = F$, that is, indices of subsets $R_j$ in \textsc{Set Cover}.
Since $|F| \le k$, $|T| \le k$, so it suffices to show that $U = \cup_{t \in T} t$.
Solution to \textsc{Beacon-Privacy-D} means that flips $F$ guarantee privacy of each $i \in B = U$.
By the observation above that it suffices to flip any query $j$ with $d_{ij} = 1$ to guarantee the privacy of $i$, $R_j$ is the subset of individuals for whom privacy is guaranteed, and, thus, $\cup_{j \in F} R_j = U$, since we must guarantee privacy of all individuals.
Since $T = F$, we have covered the universe $U$.

For the other direction, suppose that there exists a solution to \textsc{Set Cover}, $T$ with $|T| \le k$ and $U = \cup_{t \in T} t$.
Set $F = T$ and flip all queries with $j \in F$.
Since it suffices to guarantee privacy of any $i \in B$ by flipping any query $j$ with $d_{ij} = 1$, and since $R_j$ is the collection of all individuals for whom we can guarantee privacy by flipping $j$, and since $\cup_{t \in T} t = U$, by our construction this implies that privacy is guaranteed for all $i \in B$.
\section{Proof of Theorem~\ref{T:bc}}
\label{A:a3}
Recall that for a fixed-threshold attack, Beacon privacy guarantee for a given $F$ and associated indicator vector $y$, formalized in Equation~\eqref{E:privacy}, is that
\[\forall i \in B, \quad L_i(Q,x,y) = \sum_{j \in Q_1} \Delta_{ij}y_j + \eta_i \ge \theta.\]
Let $\Delta_i = \min_{j \in P_i} \Delta_{ij}$.
Then if for each $i$,
\(\Delta_{i} + \eta \ge \theta,\)
the condition above certainly follows as well, since $\sum_{j \in Q_1} \Delta_{ij}y_j \ge \Delta_i$ by definition of a \emph{Beacon-Cover}, and 
\begin{gather*}
\eta=\min_i \left(\sum_{j \in Q_1} d_{ij} \log(1-D_n^j) + \sum_{j \in Q_0} d_{ij}  \log\frac{D_n^j}{0.25 D_{n-1}^j} \right)\\
\le \min_i \left(\sum_{j \in Q_1} d_{ij} \log\frac{1-D_n^j}{1-\delta D_{n-1}^j} + \sum_{j \in Q_0} d_{ij}\log\frac{D_n^j}{\delta D_{n-1}^j} \right)\\
\le \eta_i.
\end{gather*}
Now, $\Delta_i = \min_{j \in P_i} d_{ij}(B_j - A_j)$, and since $\Delta_{ij} > 0$ for any $j \in P_i$, $d_{ij} = 1$ for any $j \in P_i$.
Consequently,
\begin{gather*}
\Delta_i = \min_{j \in P_i} (B_j - A_j) = \min_{j \in P_i} \left( \log\frac{D_n^j}{\delta D_{n-1}^j} - \log\frac{1-D_n^j}{1-\delta D_{n-1}^j}\right)\\
= \min_{j \in P_i} \left(\log D_n^j - \log \delta D_{n-1}^j - \log(1-D_n^j) \right.\\\left. + \log(1-\delta D_{n-1}^j)\right)\\
= \min_{j \in P_i} \left(\log \frac{D_n^j}{1-D_n^j} + \log\frac{1-\delta D_{n-1}^j}{\delta D_{n-1}^j }\right)\\\ge \min_{j \in P_i} \left(\log \frac{D_n^j}{1-D_n^j}\right) + \min_{j \in P_i}  \left(\log\frac{1-\delta D_{n-1}^j}{\delta D_{n-1}^j }\right)\\
\ge D_n + \min_{j \in P_i}  \left(\log\left(\frac{1}{\delta D_{n-1}^j } - 1\right)\right) \ge D_n + \log\left(\frac{1}{\delta} - 1\right),
\end{gather*}
where the last inequality follows since $D_{n-1}^j \le 1$.
Now, if $\delta \le \frac{1}{1+e^{\theta - \eta - D_n}}$,
then
\begin{align*}
\log\left(\frac{1}{\delta} - 1\right) \ge \log\left(e^{\theta - \eta - D_n}\right)
=\theta - \eta - D_n.
\end{align*}
Consequently, $\Delta_i \ge \theta - \eta$ for each $i \in B$, which is just a rearranging of the desired condition above.

\section{Algorithm - Greedy Min Beacon Cover}
\label{A:alg1}
\begin{algorithm}[h]
\SetAlgoLined
\KwInput{A set $B$ of individuals in the Beacon, a subset $P_i$ for each individual, and a collection of queries $Q$.}
\KwOutput{Subsets of queries $F \subseteq S$ to flip.}

\KwInit{$F = \emptyset$, $C = \emptyset$.}

  \While{$(B \setminus C) \ne \emptyset$} {
    Set $l = 1, T = \emptyset, N = -1$.\\
    \For{$j \in (Q \setminus F)$}{
      Set $T_j = \{i \in (B \setminus C) | j \in P_i\}$.\\
      \If{$|T_j| > N$}{
      Set $T = T_j$.\\
      Set $N = |T|$.\\
      Set $l = j$.
      }
      Set $F = F \cup l$.\\
      Set $C = C \cup T$.
    }
 }
 \caption{Greedy \textsc{Min Beacon Cover} Algorithm.}
 \label{algo:cover}
\end{algorithm}

\section{Proof of Proposition~\ref{P:onlinegreedy}}
\label{A:a4}
We prove this by induction.
For base case, note that privacy is guaranteed at $t=0$ since $L_i(\emptyset,x,y) = 0 \ge \theta$ for all $i$ and $\theta \le 0$.
Next, suppose that $L_i(Q_{t-1},x,y_{t-1}) \ge \theta$.
If $L_i(Q_{t},x,y_{t-1}) \ge \theta$, that is, we need not flip the response to the current query $q_t$, privacy is not violated at time $t$.
Suppose that $L_i(Q_{t},x,y_{t-1}) < \theta$, which means that we flip the response to query $q_t$ in the \textsc{Online Greedy Algorithm}.
Let $y_{j,t} = y_{j,t-1}$ for all $j \ne q_t$ and $y_{j,t} = 1$ for $j=q_t$.
Then
\begin{align*}
L_i(Q_t,&x,y_t) = \sum_{j \in Q_{t,1}} \Delta_{ij} y_{j,t} + \eta_i(Q_{t})\\
&=\sum_{j \in Q_{t-1,1}} \Delta_{ij} y_{j,t-1} + \eta_i(Q_{t-1}) + \eta_i(q_t) + \Delta_{i,q_t}\\
&=L_i(Q_{t-1},x,y_{t-1}) + \eta_i(q_t) + \Delta_{i,q_t}.
\end{align*}
Now, $\eta_i(q_t) = d_{i,q_t}(x_{q_t}A_{q_t} + (1-x_{q_t})B_{q_t})$, while $\Delta_{i,q_t} = d_{i,q_t}(B_{q_t} - A_{q_t})$.
Moreover, recall that if $\delta < 0.25$, $B_j > 0$ and $A_j < 0$ for all queries $j$.
Since $L_i(Q_{t},x,y_{t-1}) < \theta$, it must be that $d_{i,q_t} = 1$, since otherwise $\eta_i(q_t)=0$, and $x_{q_t} = 1$, since otherwise $\eta_i(q_t)>0$.
Thus, $\eta_i(q_t) = A_{q_t}$ and $\Delta_{i,q_t} = B_{q_t} - A_{q_t}$.
Consequently, $\eta_i(q_t) + \Delta_{i,q_t} = A_{q_t} + B_{q_t} - A_{q_t} = B_{q_t} > 0 \ge \theta$.
Since this holds for every individual $i$ and time $t$, privacy against threshold attacks is guaranteed for all individuals and query sequences.

\section{Proof of Proposition~\ref{P:advftatt}}
\label{A:a5}
Fix $i \in B$.
We begin by unpacking the LRT score resulting from a flipping strategy $y$ in Equation~\eqref{E:advprivacysimple} (since $\theta$ is fixed, that is the only thing affected by the choices of queries):
\begin{gather*}
\min_{Q_i \subseteq S} L_i(Q_i,x,y) = \sum_{j \in Q_{i,1}} \Delta_{ij} y_j + \sum_{j \in Q_{i,1}} d_{ij} A_j  \\ + \sum_{j \in Q_{i,0}} d_{ij} B_j.
\end{gather*}
First, observe that since $B_j > 0$ by our assumption that $\delta < 0.25$ and from (the proof of the) Proposition~\ref{lemma_flip}, query set $S_0$ (i.e., those with $x_j = 0$) would not be included, since they can only increase the LRT statistic.
Similarly, none of the queries with $d_{ij} = 0$ will be included since these do not contribute to the LRT statistic.
Consequently, $Q_i \subseteq P_i(S)$.
Moreover, since $B_j > A_j$ under the same assumptions, none of the terms with $y_j = 1$ are included.
Consequently, $Q_i \subseteq P_i(S) \setminus F$.
Moreover, since $A_j < 0$, \emph{all} queries in $P_i(S)\setminus F$ will be included.
\section{Adaptive Threshold, Batch Setting}
Fig.~\ref{fig:adaptive_c} presents a comparison of privacy achieved by \textsc{MIG} and the various baselines, Fig.~\ref{fig:adaptive_d} and~\ref{fig:adaptive_e} present ROC curves for the adaptive attack in the batch setting. Note that in this case, a lower area under the curve is better, as the ROC curve corresponds to attack success. 

Our approach, \textsc{MIG} outperforms all baselines, except \textsc{SF-M} which also achieves perfect privacy in the batch setting, although with significantly lower utility than \textsc{MIG}, as can be observed in Fig.~\ref{fig:adaptive_a}. Also note that the maximum false positive rate up until which the area under the ROC curve (AUC) remains zero for \textsc{MIG} and \textsc{SF-M} corresponds to the percentage of the total population for which the solution is computed ($20\%$ of individuals not in the beacon corresponds to $10\%$ of the total population), beyond which the AUC is non-zero for all approaches (the plot line corresponding to \textsc{SF-M} has a slight non-zero slope between FPR=$0.1$ and FPR=$0.2$ in Fig.~\ref{fig:adaptive_d}).
\label{A:a6}
\begin{figure}[h]
  \centering
  \captionsetup[subfigure]{justification=centering}
  \hspace*{3em}
  \begin{subfigure}[t]{0.47\linewidth}
      \includegraphics[width=\textwidth]{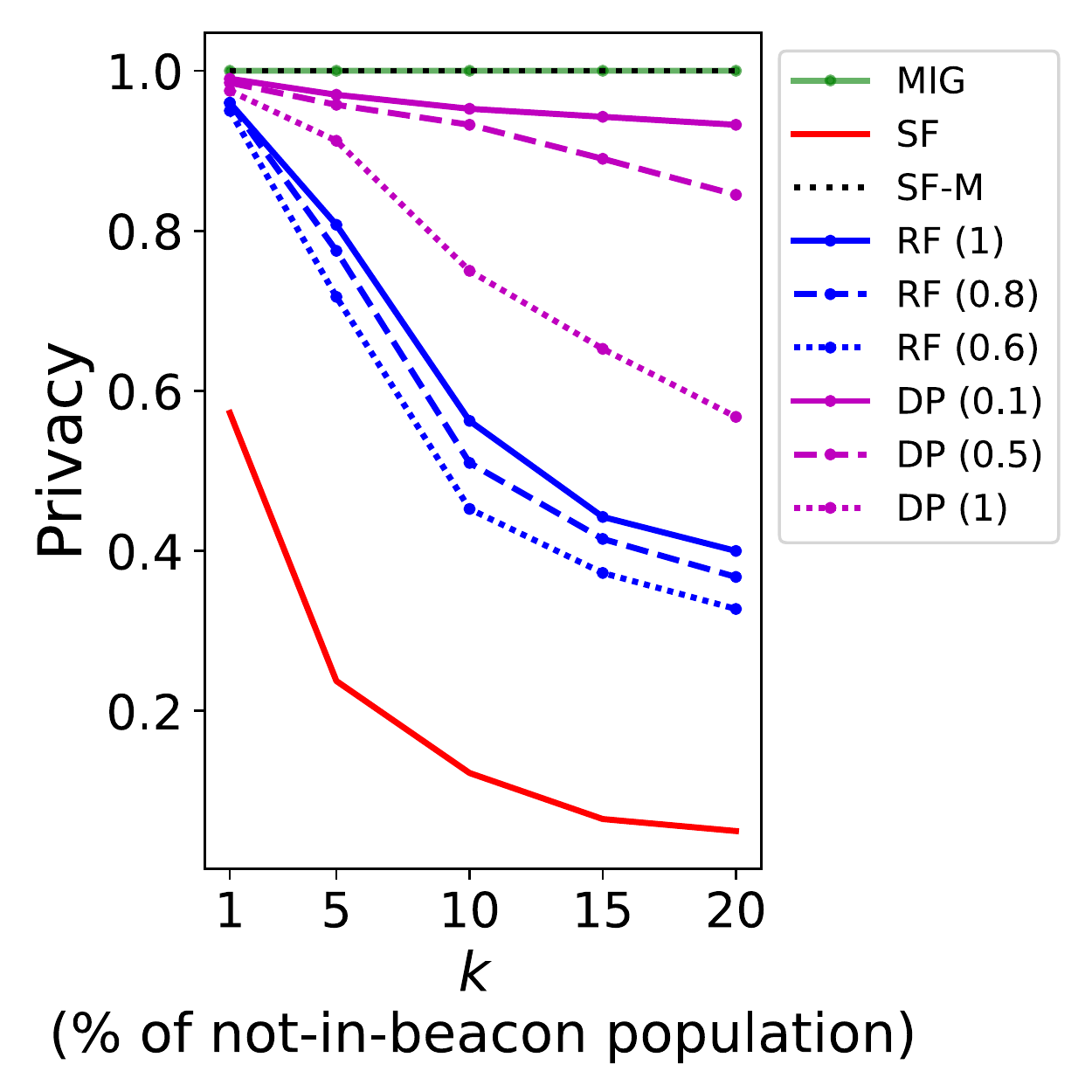}
      \caption{Privacy comparison}
      \label{fig:adaptive_c}
  \end{subfigure}
  \newline
  \begin{subfigure}[t]{0.47\linewidth}
      \centering
      \includegraphics[width=\textwidth]{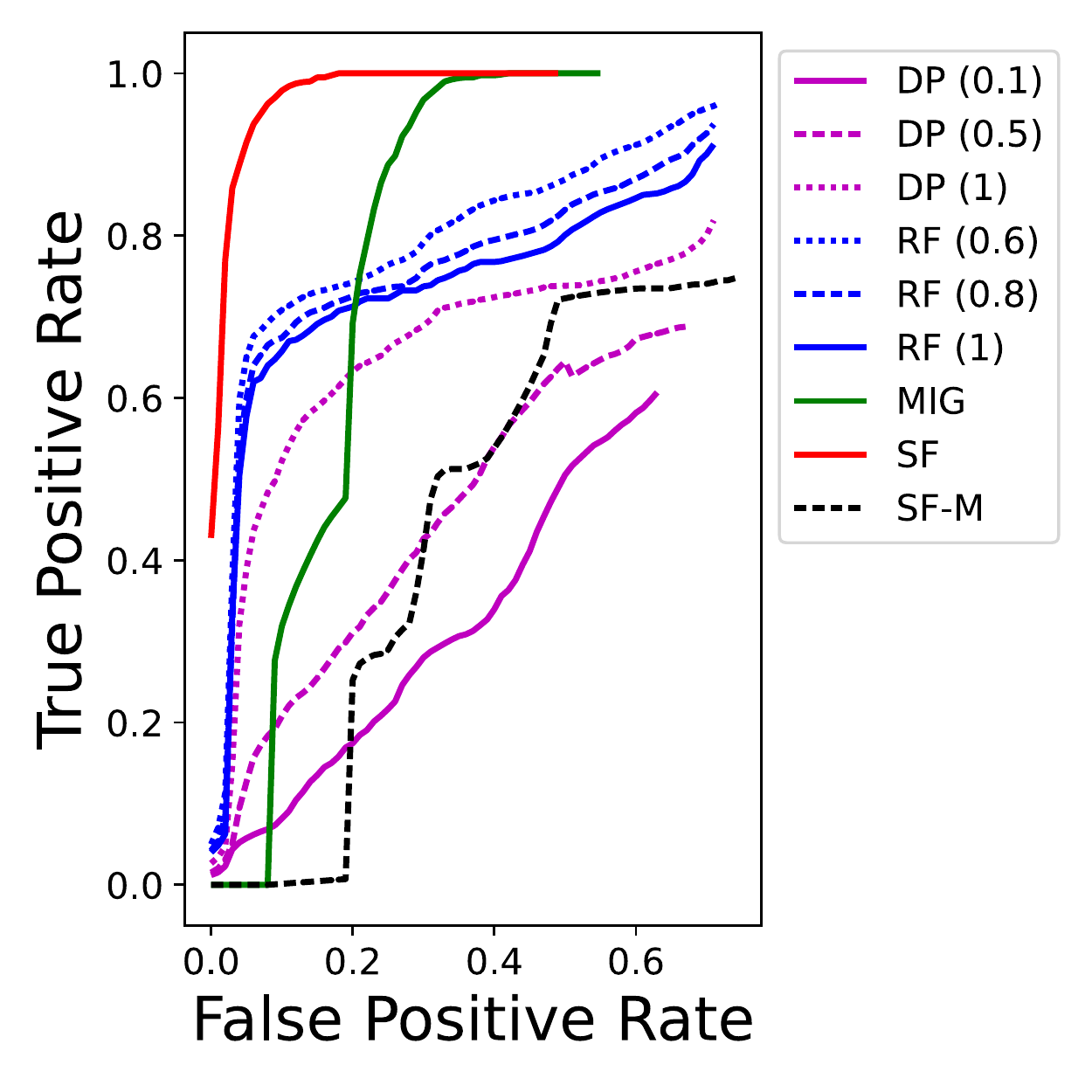}
      \caption{ROC Curves for the Adaptive Batch Setting ($k$=20)}
      \label{fig:adaptive_d}
  \end{subfigure}
  \begin{subfigure}[t]{0.47\linewidth}
      \centering
      \includegraphics[width=\textwidth]{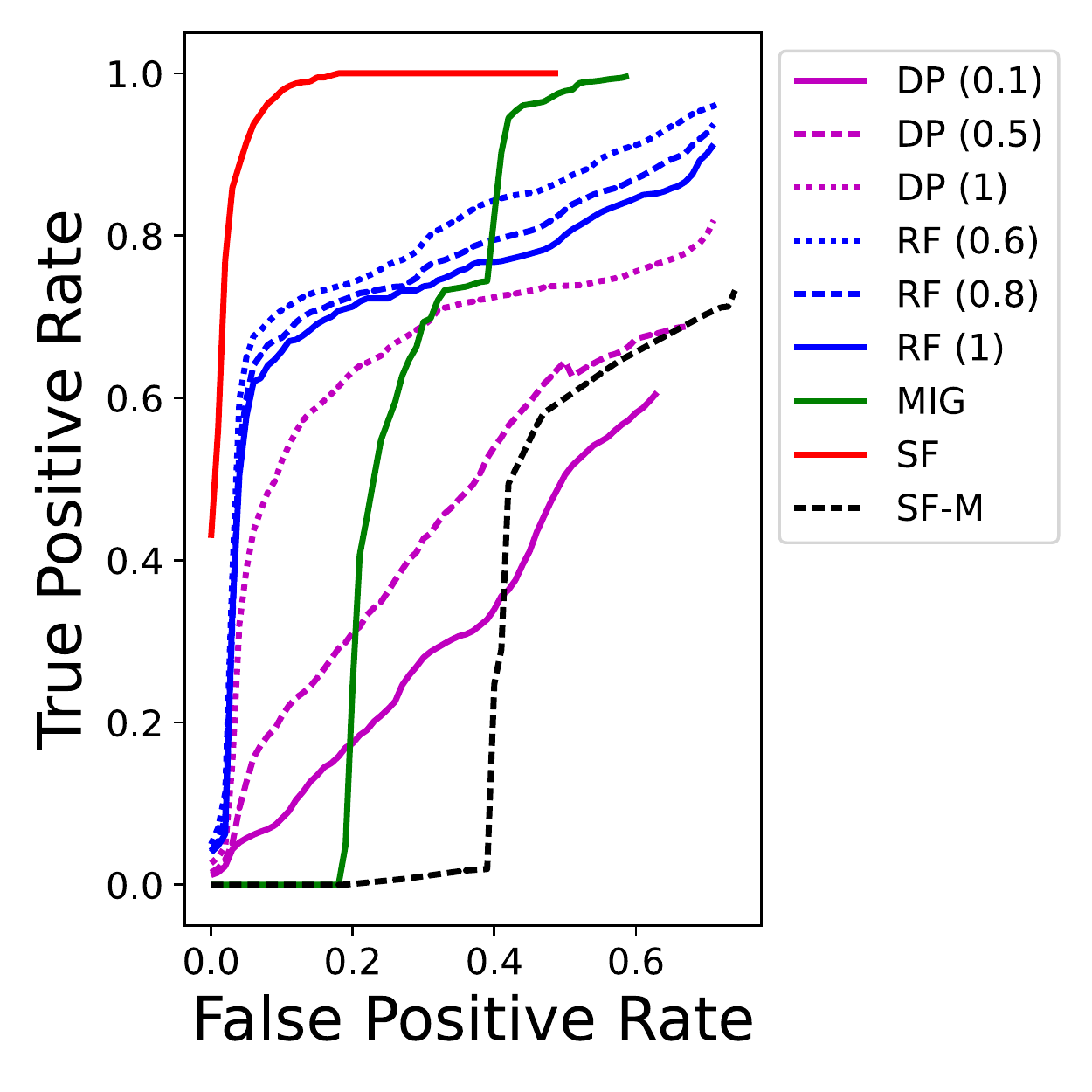}
      \caption{ROC Curves for the Adaptive Batch Setting ($k$=40)}
      \label{fig:adaptive_e}
  \end{subfigure}
    \caption{Performance in the adaptive attack batch setting.}
  \label{fig:adaptive_priv_roc}
\end{figure}
\section{Adaptive Threshold, Authenticated Online Setting}
\label{A:a7}
Next, we compare the performance of the various methods using ROC curves in the authenticated online setting with the adaptive threshold model. Figs.~\ref{fig:auth_online_k1_100k}, \ref{fig:auth_online_k1_500k} and~\ref{fig:auth_online_k1_1.3M} present ROC curves when $100000$, $500000$ and $1.3$ million SNVs are queried respectively, using $k=1$; and Figs.~\ref{fig:auth_online_k10_100k}, \ref{fig:auth_online_k10_500k} and~\ref{fig:auth_online_k10_1.3M} present corresponding results for $k=10$. 

When $k=1$, the variant of our approach in this setting (\textsc{OG}) outperforms all baselines except \textsc{DP} with $\epsilon=0.1$ when $100000$ SNVs are queried, and \textsc{DP} with $\epsilon=0.5$ and $\epsilon=0.1$ when the number of SNVs queried is increased to $500000$ and $1.3$ million. However, \textsc{DP} flips a significantly larger number of SNVs when compared to \textsc{MIG}, as can be observed in Fig.~\ref{fig:authOnline_ada_a}. Also note that at $k=1$, \textsc{DP} does violate privacy more often than \textsc{MIG}, as can be seen in Fig.~\ref{fig:authOnline_ada_c}. When $k=10$, \textsc{OG} outperforms all baselines, except \textsc{DP} with $\epsilon=0.1$, although the performance of the two approaches are closer in this case when compared to the case where $k=1$. The performance of \textsc{SF-M} also shows significant improvement relative to the $k=1$ case.

\section{Adaptive Threshold, Unauthenticated Online Setting}
\label{A:a8}
Fig.~\ref{fig:unauth_online_k10_ROCs} presents ROC curves comparing the performance of the various baselines to our variant (\textsc{OMIG}) in the unauthenticated online setting with an adaptive threshold attack, when $100000$, $500000$ and $1.3$ million SNVs are queried. In all three cases, \textsc{DP} outperforms \textsc{OMIG}, and \textsc{SF} achieves very similar performance to \textsc{OMIG}. However, both \textsc{DP} and \textsc{SF} provide lower utility when compared to \textsc{OMIG}, as can be seen from Fig.~\ref{fig:unauthOnline_ada}.
\begin{figure}[t]
  \centering
  \captionsetup[subfigure]{justification=centering}
  \begin{subfigure}[t]{0.47\linewidth}
      \centering
      \includegraphics[width=\textwidth]{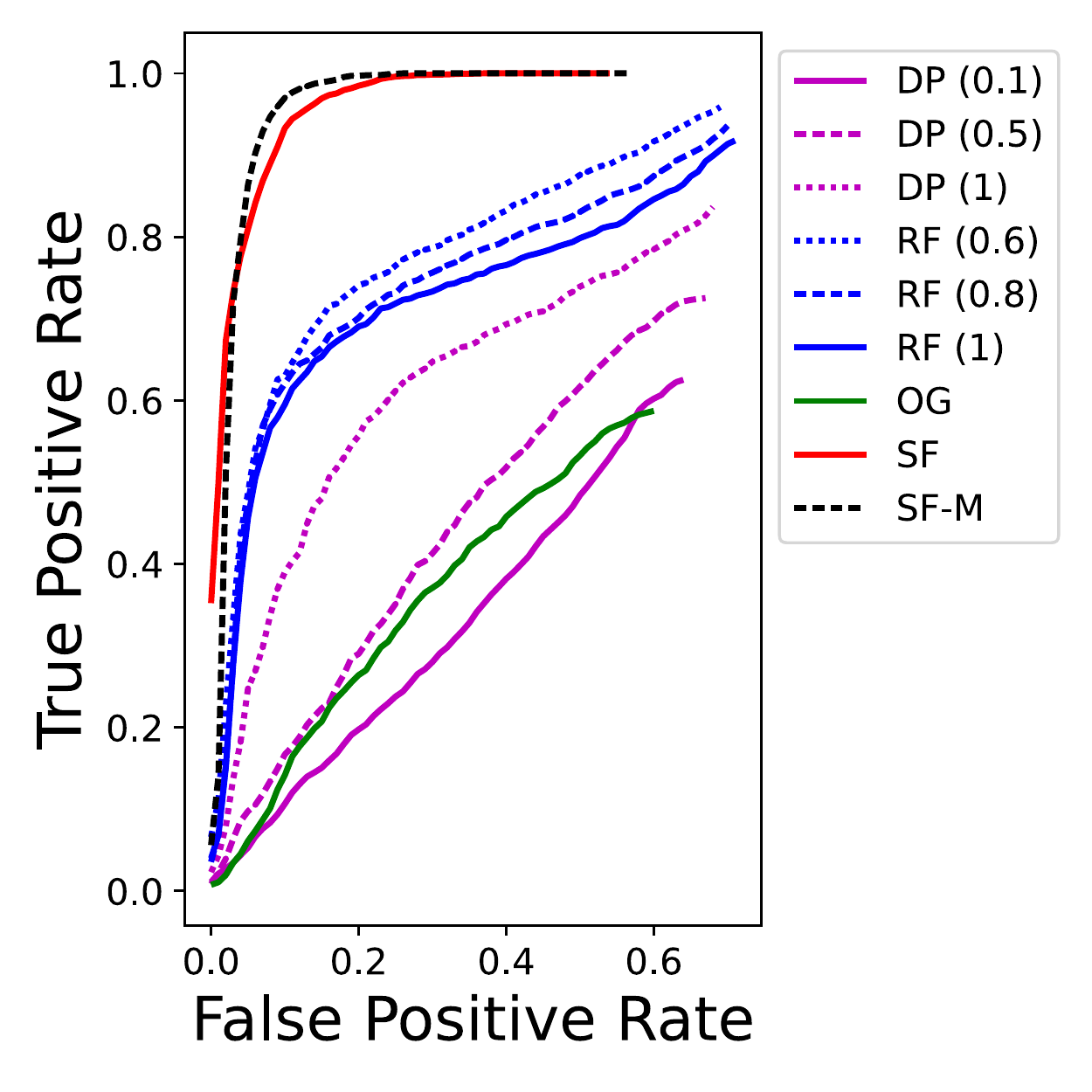}
      \caption{$k$=1, 100K SNVs queried}
      \label{fig:auth_online_k1_100k}
  \end{subfigure}
  \begin{subfigure}[t]{0.47\linewidth}
      \centering
      \includegraphics[width=\textwidth]{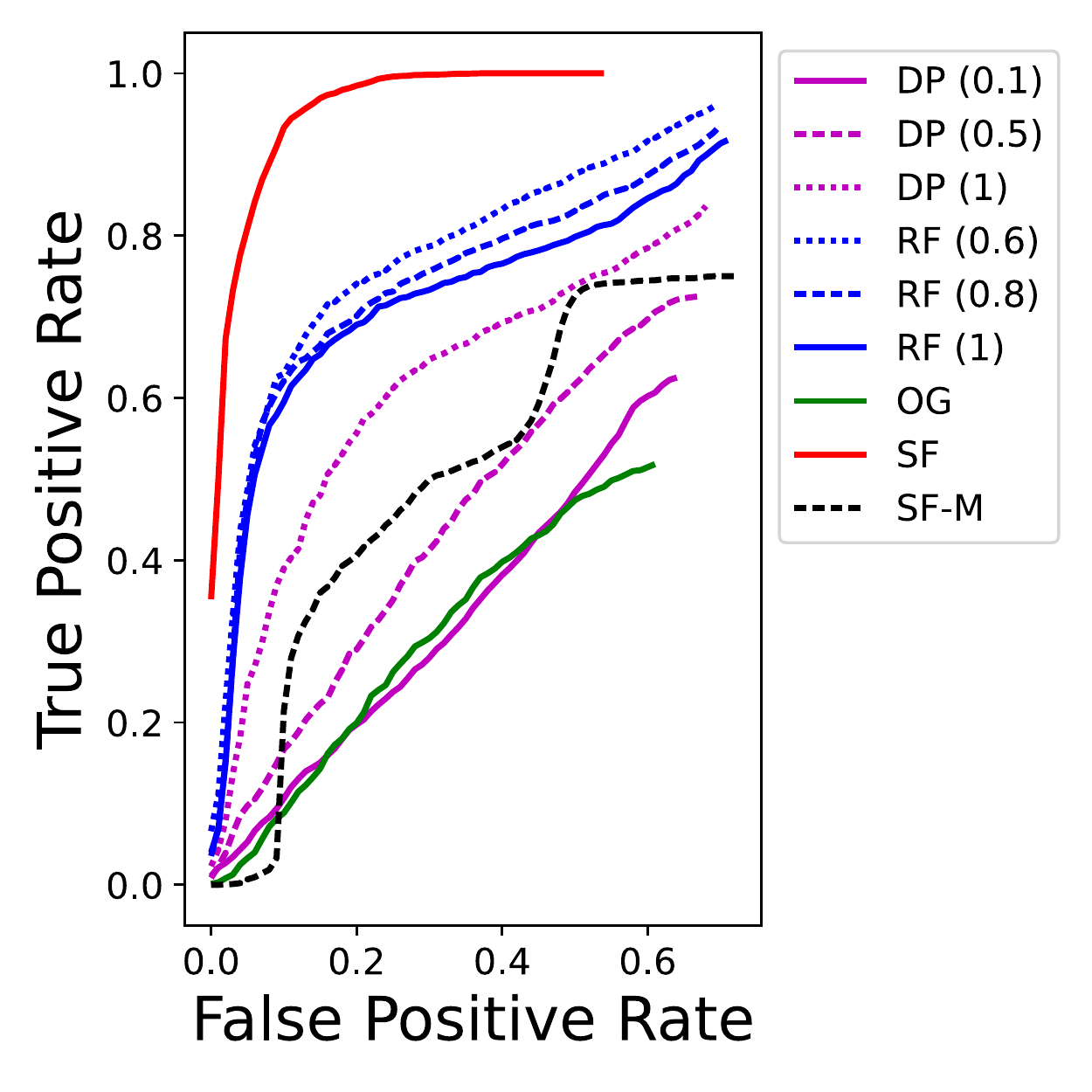}
      \caption{$k$=10, 100K SNVs queried}
      \label{fig:auth_online_k10_100k}
  \end{subfigure}
  \newline
  \begin{subfigure}[t]{0.47\linewidth}
      \centering
      \includegraphics[width=\textwidth]{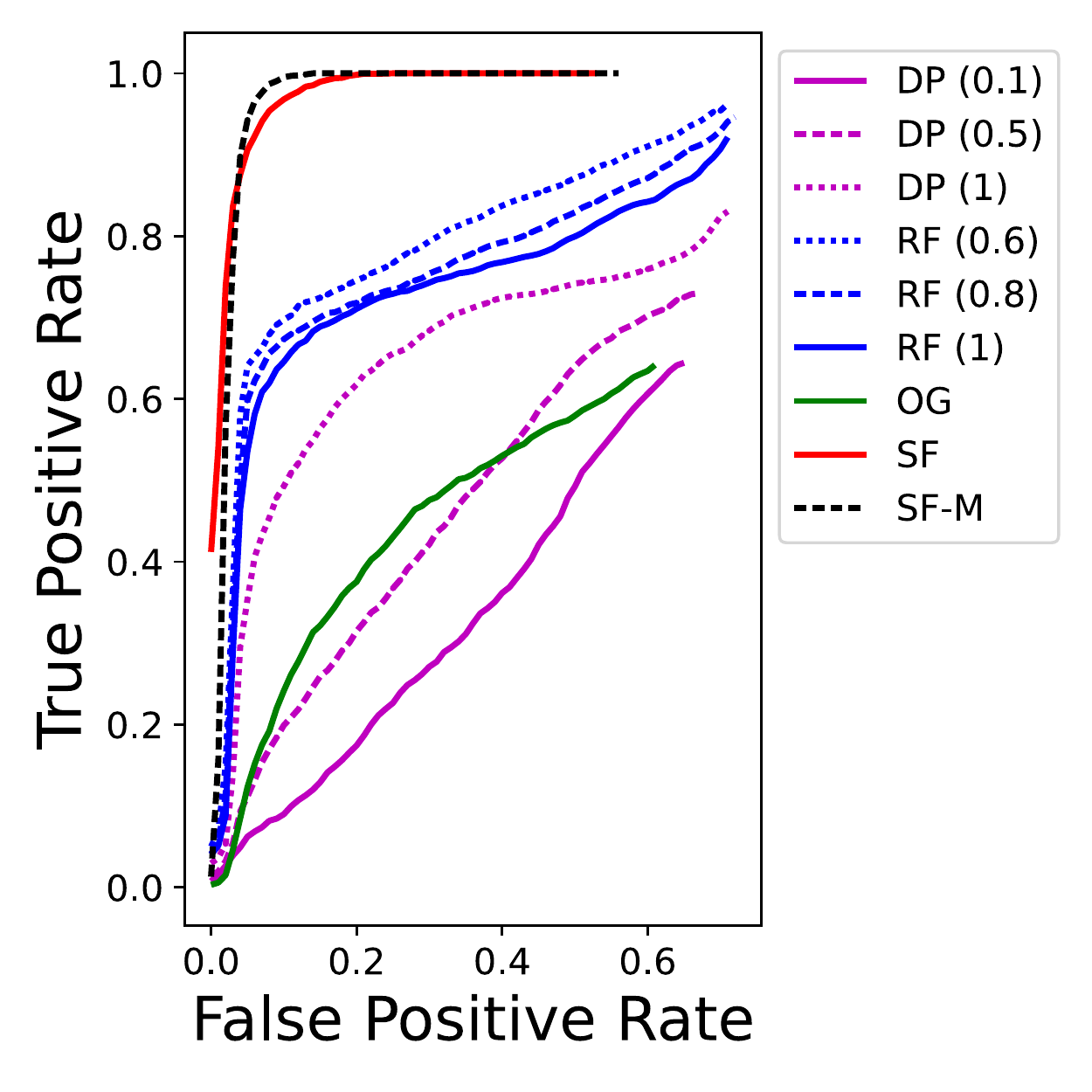}
      \caption{$k$=1, 500K SNVs queried}
      \label{fig:auth_online_k1_500k}
  \end{subfigure}
  \begin{subfigure}[t]{0.47\linewidth}
      \centering
      \includegraphics[width=\textwidth]{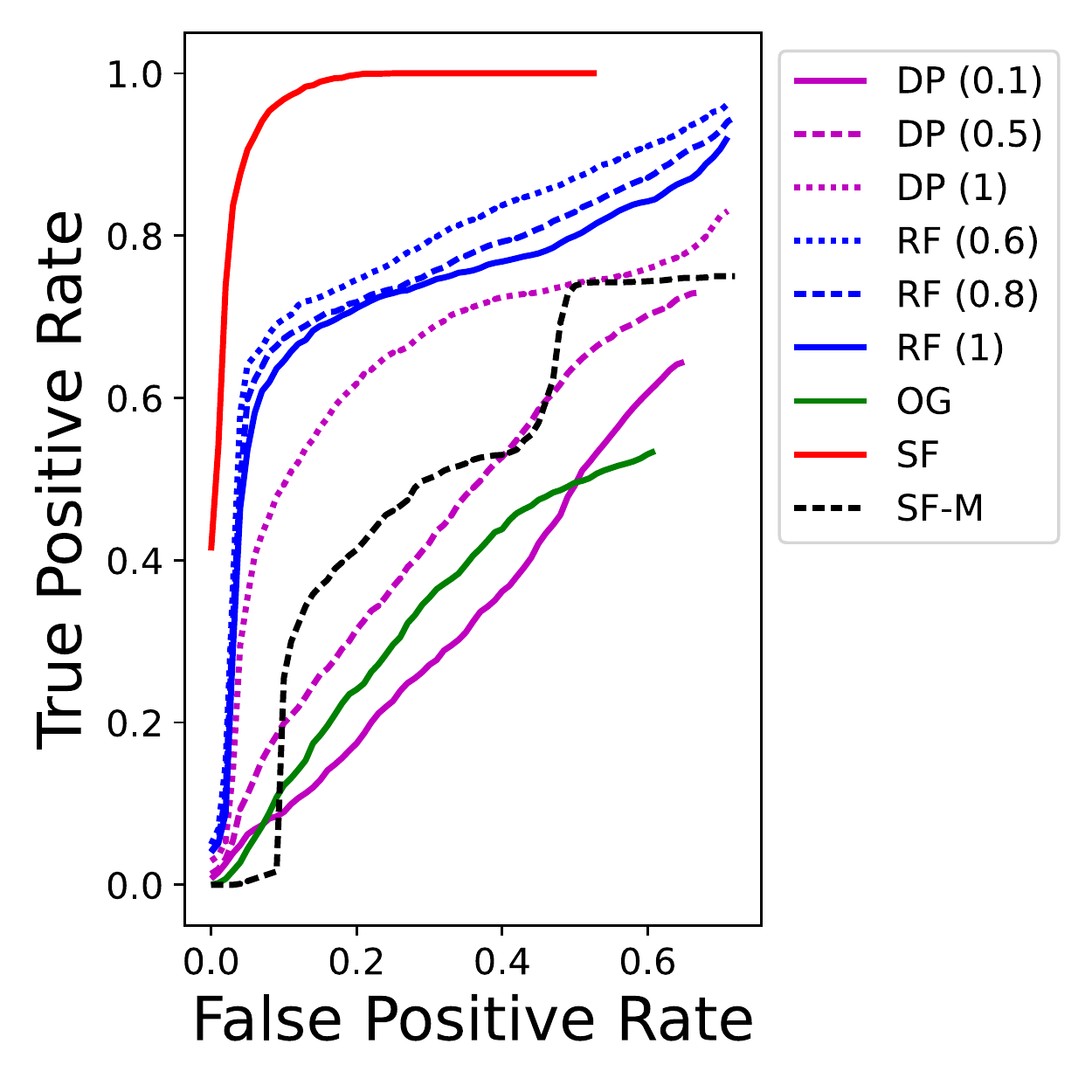}
      \caption{$k$=10, 500K SNVs queried}
      \label{fig:auth_online_k10_500k}
  \end{subfigure}
  \newline
  \begin{subfigure}[t]{0.47\linewidth}
      \centering
      \includegraphics[width=\textwidth]{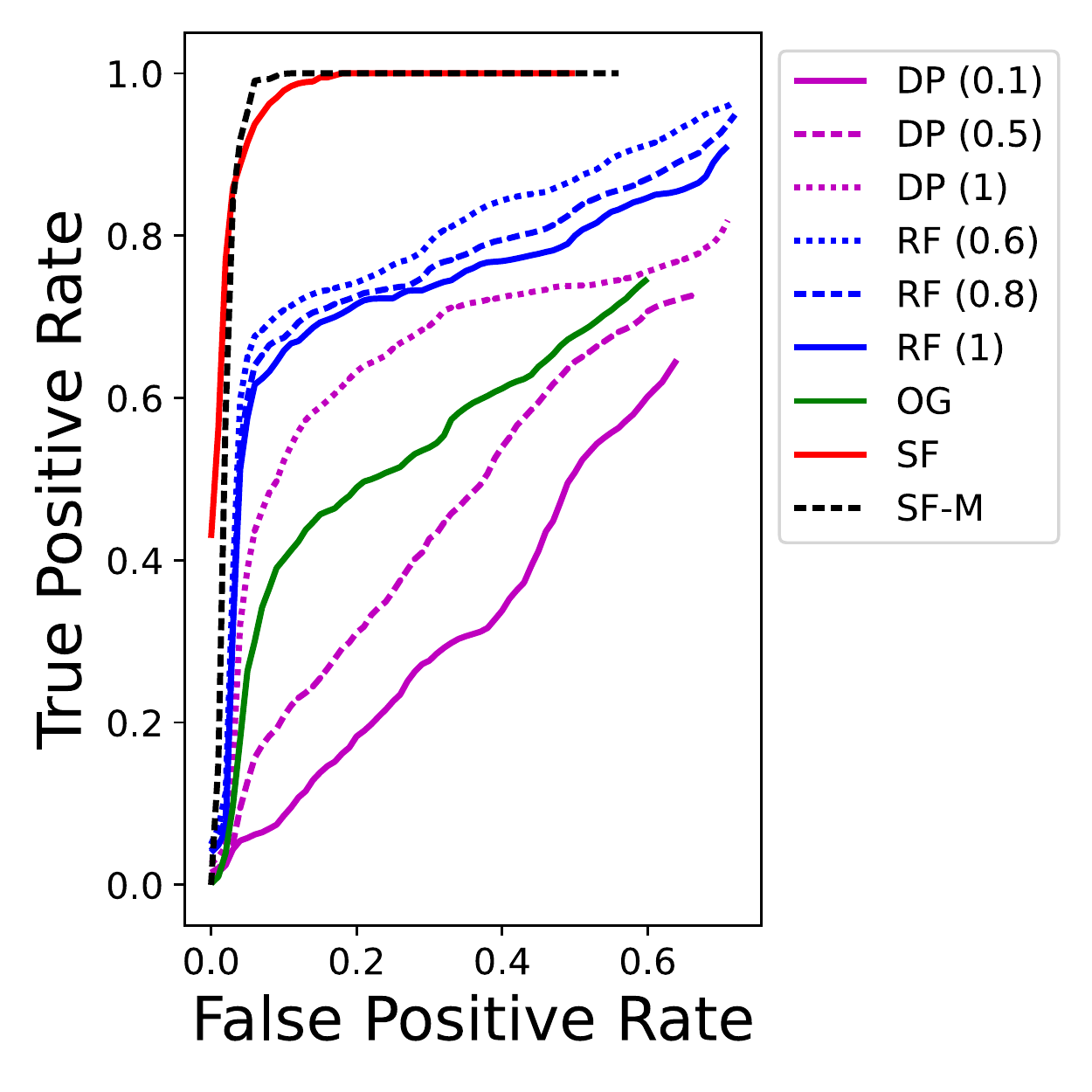}
      \caption{$k$=1, 1.3M SNVs queried}
      \label{fig:auth_online_k1_1.3M}
  \end{subfigure}
  \begin{subfigure}[t]{0.47\linewidth}
      \centering
       \includegraphics[width=\textwidth]{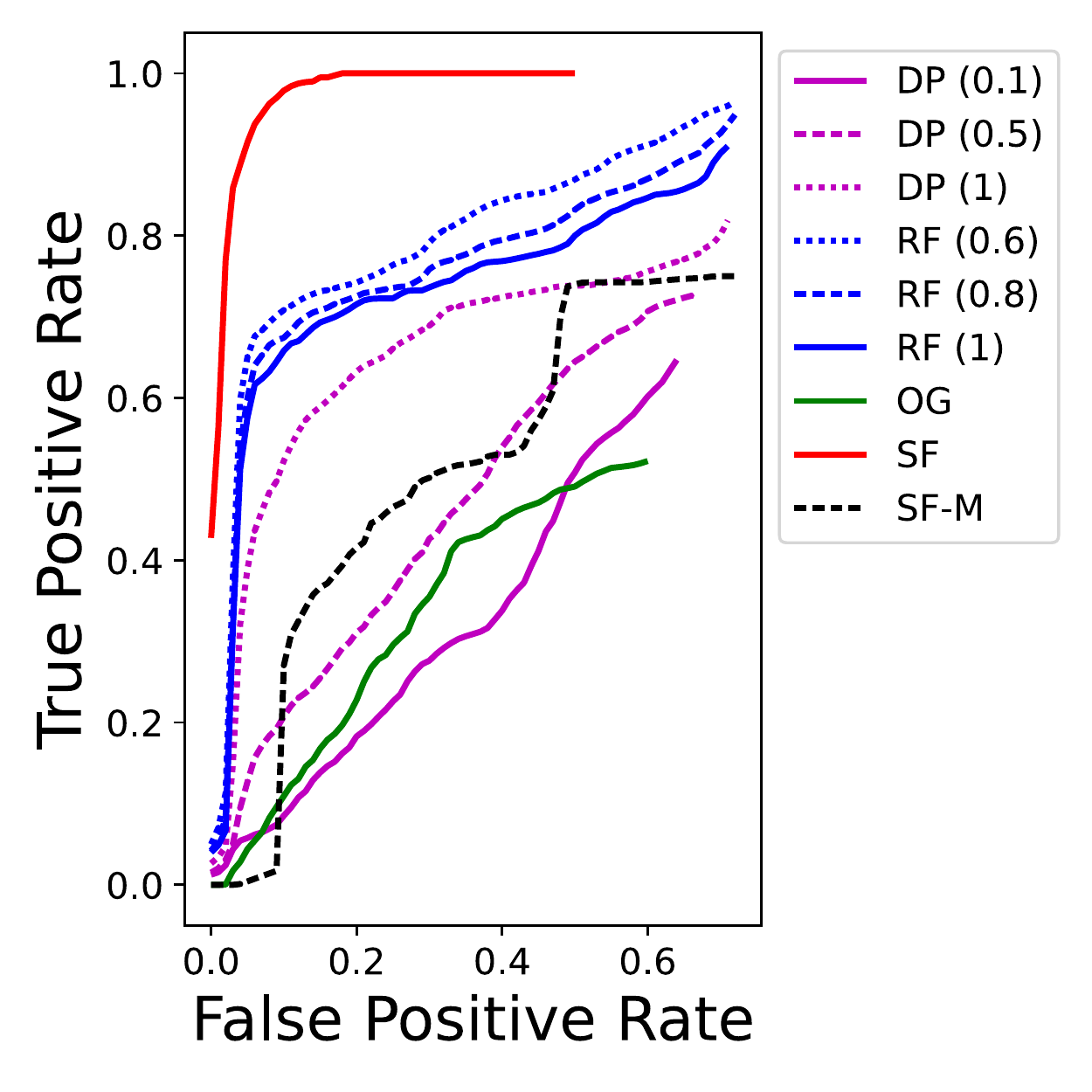}
      \caption{$k$=10, 1.3M SNVs queried}
      \label{fig:auth_online_k10_1.3M}
  \end{subfigure}
    \caption{ROC curves, authenticated online setting.}
  \label{fig:auth_online_k1_k10_ROCs}
\end{figure}
\begin{figure}[h]
  \centering
  \captionsetup[subfigure]{justification=centering}
  \begin{subfigure}[t]{0.47\linewidth}
      \centering
      \includegraphics[width=\textwidth]{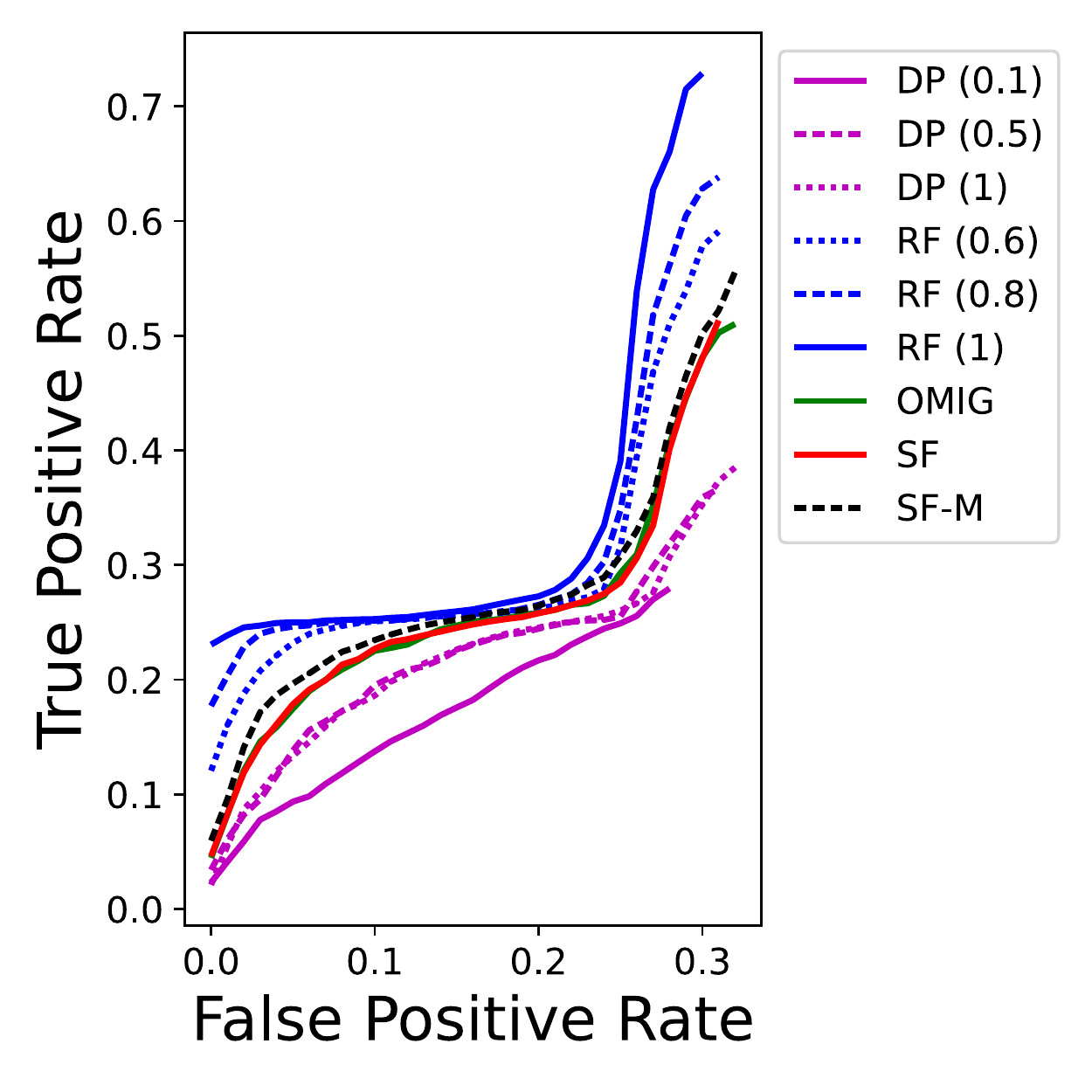}
      \caption{$k$=10, 100K SNVs queried}
      \label{fig:unauth_online_k10_100k}
  \end{subfigure}
  \begin{subfigure}[t]{0.47\linewidth}
      \centering
      \includegraphics[width=\textwidth]{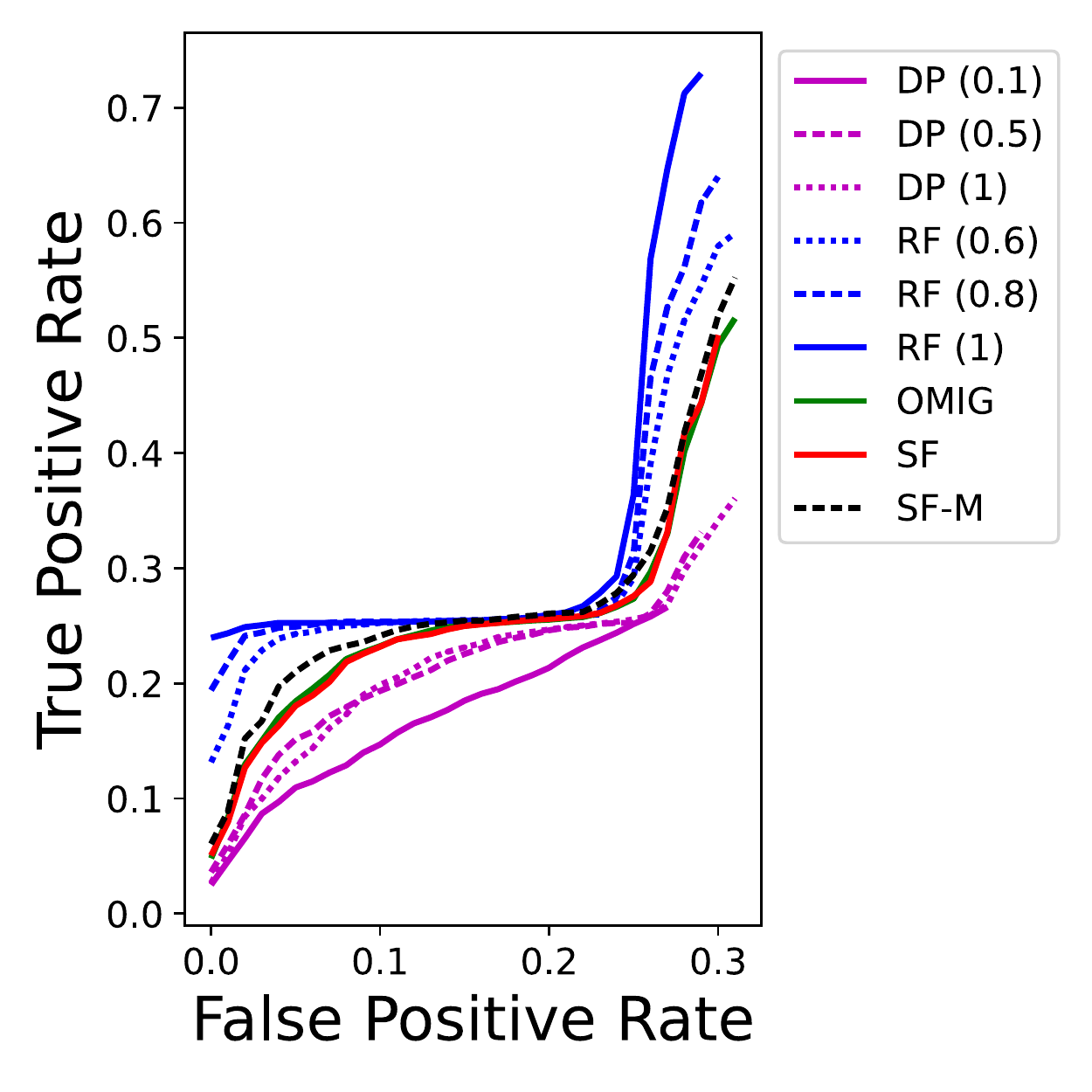}
      \caption{$k$=10, 500K SNVs queried}
      \label{fig:unauth_online_k10_500k}
  \end{subfigure}
  \newline
  \begin{subfigure}[t]{0.47\linewidth}
      \centering
      \includegraphics[width=\textwidth]{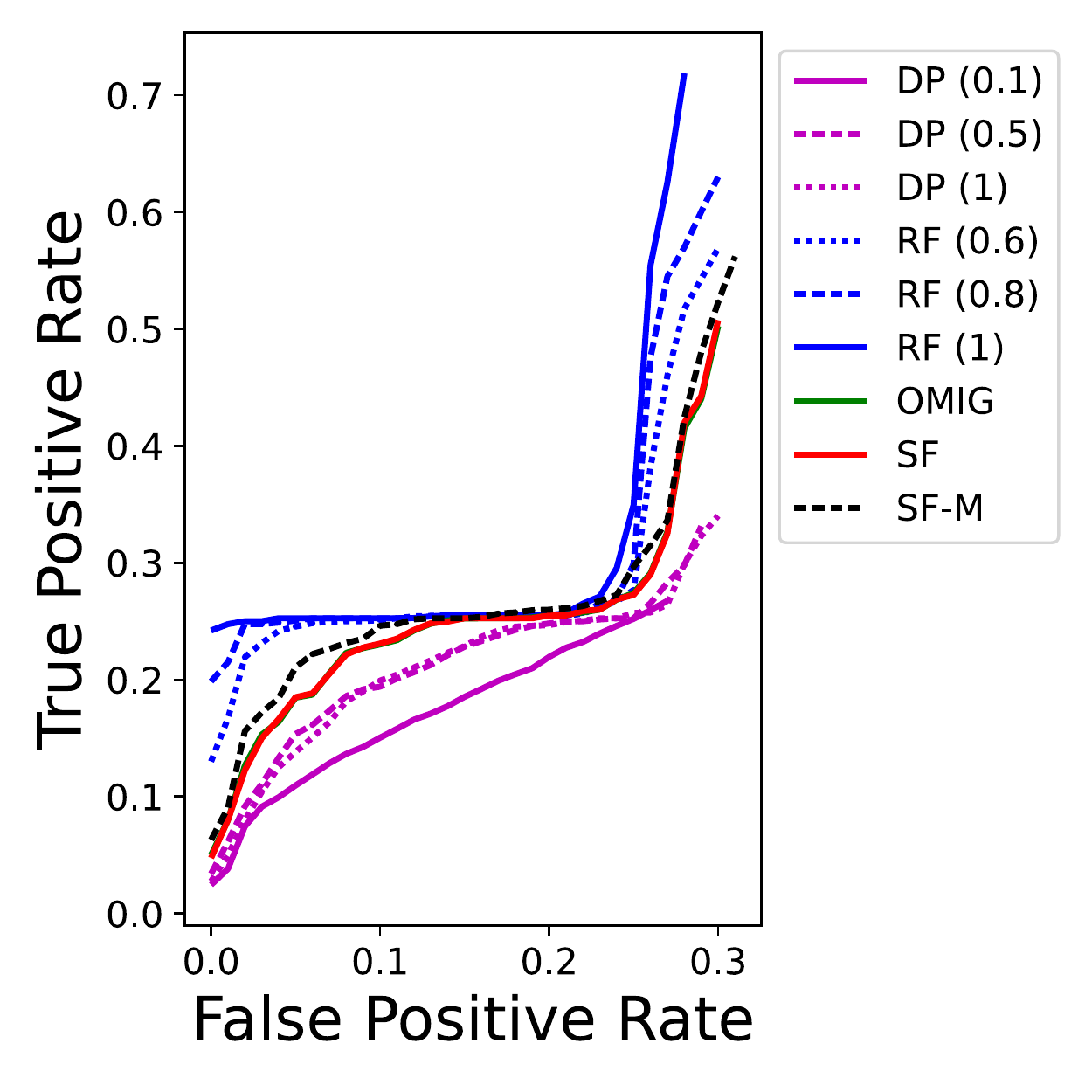}
      \caption{$k$=10, 1.3M SNVs queried}
      \label{fig:unauth_online_k10_1.3M}
  \end{subfigure}
    \caption{ROC curves, unauthenticated online setting.}
  \label{fig:unauth_online_k10_ROCs}
\end{figure}

\end{document}